\documentclass[aps,pra,reprint,superscriptaddress,nofootinbib,floatfix,longbibliography]{revtex4-2}

\usepackage{graphicx}
\usepackage{amsmath,amssymb,amsthm,mathtools,bm}
\usepackage{algorithm}
\usepackage{algpseudocode}
\usepackage{booktabs}
\usepackage{float}
\usepackage{enumitem}
\usepackage{url}
\usepackage{xcolor}
\usepackage{hyperref}
\graphicspath{{figures/}}

\hypersetup{colorlinks=true,linkcolor=blue,citecolor=blue,urlcolor=blue}

\newtheorem{definition}{Definition}[section]
\newtheorem{theorem}[definition]{Theorem}
\newtheorem{proposition}[definition]{Proposition}

\newtheorem{remark}[definition]{Remark}

\DeclareMathOperator{\Tr}{Tr}
\DeclareMathOperator{\rank}{rank}

\DeclareMathOperator{\spec}{spec}

\newcommand{\R}{\mathbb R}
\newcommand{\C}{\mathbb C}

\newcommand{\EE}{\mathbb E}
\newcommand{\Herm}{\mathrm{Herm}}

\newcommand{\MM}{\mathcal M}
\newcommand{\ket}[1]{\lvert #1\rangle}
\newcommand{\bra}[1]{\langle #1\rvert}
\newcommand{\ketbra}[1]{\lvert #1\rangle\!\langle #1\rvert}
\newcommand{\braket}[2]{\langle #1\mid #2\rangle}

\newcommand{\calT}{\mathcal T}

\newcommand{\calJ}{\mathcal J}

\numberwithin{equation}{section}

\begin{document}
\title{Quantum Spectral Anomaly Detection}
\author{Yewei Yuan}
\affiliation{Global College, Shanghai Jiao Tong University, Shanghai 200240, China}

\author{Michele Minervini}
\affiliation{School of Electrical and Computer Engineering, Cornell University, Ithaca, New York 14850, United States}

\author{Mark M. Wilde}
\affiliation{School of Electrical and Computer Engineering, Cornell University, Ithaca, New York 14850, United States}

\author{Nana Liu}
\affiliation{Global College, Shanghai Jiao Tong University, Shanghai 200240, China}
\affiliation{Institute of Natural Sciences, Shanghai Jiao Tong University, Shanghai 200240, China}
\affiliation{School of Mathematical Sciences, Shanghai Jiao Tong University, Shanghai 200240, China}
\date{\today}

\begin{abstract}
A core task in quantum anomaly detection is to compute an anomaly score that quantifies how strongly a test quantum state deviates from a given quantum dataset assumed to be normal.
Classically, principal component analysis (PCA) for centered data computes the anomaly score by evaluating the test sample relative to the subspace spanned by the selected leading eigenvectors.
However, for quantum data that lack a standard centering, explicitly recovering principal eigenvectors, constructing full Gram matrices, or loading quantum-random-access-memory-style data can be more costly than estimating the anomaly score itself.
To avoid these costs, we propose Quantum Spectral Anomaly Detection (QSPADE), which computes PCA-like anomaly scores directly from the spectrum of the average state of the normal dataset.
By replacing hard PCA rank selection with a smooth, temperature-controlled spectral threshold, QSPADE makes near-threshold spectral components contribute partially to the anomaly score.
This makes the score vary continuously rather than jump when a borderline component is included or excluded, and makes it less sensitive to noise or arbitrary hard cutoffs near the threshold.
In the zero-temperature limit, QSPADE recovers the hard-projector PCA score.
The proposed measurement-based quantum detector can be calibrated with a sample complexity independent of the data dimension. 
Numerical simulations show that QSPADE behaves like kernel-PCA on encoded classical data and detects changes across a transverse-field Ising transition without predefined order parameters.
Consequently, QSPADE gives an efficient framework for both quantum-kernel anomaly detection on encoded classical data and the monitoring of quantum-native systems where diagnostic observables are unknown.
\end{abstract}

\maketitle


\section{Introduction}
\label{sec:intro}
Anomaly detection is a one-class learning problem: given training examples that are assumed to be normal, we want to assign a new test object an anomaly score, with larger values indicating stronger deviation from the normal data.
Principal component analysis (PCA)~\cite{Jolliffe2002} gives a spectral way to construct such scores. 
After centering the normal data, the sample covariance matrix identifies directions of variation: large eigenvalues correspond to directions strongly represented in the normal data, while small eigenvalues correspond to weakly represented directions. The large-eigenvalue subspace is the principal subspace. 
Relative to this subspace, PCA produces two anomaly scores. First, it can have a large residual weight outside this subspace, as measured by the residual statistic $Q$. 
Second, it can lie inside the subspace but have unusually large weight along a low-variance direction, as measured by Hotelling's $\mathcal T^2$ statistic~\cite{Hotelling1931,JacksonMudholkar1979}.

For either method, the desired output is therefore a single scalar anomaly score. This distinction is especially important for quantum data. 
When the data are naturally quantum, such as states produced by a quantum device, or when classical data are already represented as quantum feature states, the task is to estimate this scalar score without requiring full spectral decomposition information.
However, the standard PCA workflow is not naturally matched to this goal. 
Classically, one constructs and diagonalizes a covariance or Gram matrix for centered data, chooses a target rank according to a retained-variance criterion, stores the corresponding eigenvectors, and projects each test point onto them.
For quantum data, such tasks may solve a harder problem than anomaly detection itself.

These extra costs are also visible in existing quantum anomaly detection methods~\cite{Corli2025QMLADReview}.  Quantum kernel PCA and one-class SVM~\cite{LiuRebentrost2018} use quantum kernels to obtain nonlinear feature representations, but rely on a pairwise-kernel representation and quantum-random-access-memory-style access assumption.  
Quantum PCA and quantum linear-algebra algorithms provide direct spectral access to density operators or block-encoded matrices~\cite{Lloyd2014,Gilyen2019}. However, standard qPCA usually aims at extracting or projecting onto principal spectral information, rather than directly calibrating anomaly scores. 
For quantum native data, reproducing standard PCA also requires treating centering explicitly, since no classical feature vectors or centered Gram matrix are directly available. 
Variational one-class models and quantum autoencoders provide another route~\cite{ParkHuhPark2023VQOCC,OhPark2024QSVDD,Kottmann2021VQAD}, but they replace reconstruction by trainable circuits, which introduce optimization loops and trainability concerns such as barren plateaus~\cite{McClean2018Barren}.
Recent works show that full eigenvector recovery is unnecessary when PCA score is the only target output~\cite{HossainBhattacharjee2026FSPA,Yuan2026QPCA}. They reduce the complexity and sensitivity caused by eigenvalue resolution in standard PCA.
Following this development, we seek a simple end-to-end quantum anomaly detection method avoiding the extra reconstruction, matrix-access, rank-selection, and optimization cost.

In this paper, we propose \emph{Quantum Spectral Anomaly Detection} (QSPADE), a measurement-based approach for computing PCA-like anomaly scores directly from the spectrum of the average state of the normal dataset. 
Given $N\in\mathbb{N}$ normal states $\rho_i$, the average state is a mixed state
\begin{equation*}
    C\coloneqq\bar\rho=N^{-1}\sum_i\rho_i.
\end{equation*}   
QSPADE then defines the regularized measurement-based spectral detector
\begin{equation*}
M^f_{\mu,T}\coloneqq f\!\left(\frac{C-\mu I}{T}\right),
\end{equation*}
where $f$ is a monotone smooth response function taking values in $[0,1]$ (so that $M^f_{\mu,T}$ is a measurement operator), $\mu$ is a scalar threshold corresponding to the prescribed acceptance level, and the temperature $T>0$ controls how sharply the response changes near the threshold. 
Specifically, $f$ responds to the comparison results between the eigenvalues of $C$ and the threshold $\mu$. In the limit $T\to0$, this response becomes a step function, recovering the hard-projector PCA score.
Thus, high-eigenvalue spectral components of $C$ contribute little to the residual score, low-eigenvalue components contribute strongly, and near-threshold components contribute fractionally.
For a test state $\sigma$, the probability $1-\Tr(M^f_{\mu,T}\sigma)$ corresponds to the final residual score. 
The threshold $\mu$ is calibrated by matching the average acceptance probability $\Tr(M^f_{\mu,T}C)$ to a target retained mass, rather than choosing how many eigenvectors to keep.
The calibration requires $O(\eta^{-2}\log\delta^{-1})$ sample complexity, independent of the data dimension, where $\eta >0$ and $\delta\in (0,1)$ are the target precision and failure probability, respectively.

This work contributes an efficient quantum anomaly detection method that focuses on directly evaluating target anomaly scores from multiple copies of given quantum data, bypassing redundant costs. 
We show how classical $Q$ and Hotelling's $\mathcal{T}^{2}$ scores translate to a quantum system.
We provide corresponding quantum implementations using the measurement-based spectral detector, including both qumode and qubit circuits for different anomaly scores. 
The qubit approach is natural for single-threshold residual score calculation, while the qumode approach excels at resolving how different spectra contribute to the anomaly score.
We also identify the distinct calibration resources required by these methods.
Finally, we validate the framework numerically: QSPADE successfully recovers the nonlinear decision geometries, and it robustly tracks quantum phase transitions in the transverse-field Ising model without requiring any prior specification of order parameters. The results show that QSPADE is effective for both quantum-embedded classical data and quantum-native data.

The rest of this paper is organized as follows.  Section~\ref{sec:classical-background} reviews the classical PCA monitoring target.  Section~\ref{sec:nominal-support} defines nominal spectral support for quantum data and states the hard spectral detector.  Section~\ref{sec:regularized-detector} introduces the regularized spectral detector.  Section~\ref{sec:qsad-statistics} defines the QSPADE $Q$ and $\calT^{2}$ statistics and gives their sharp-limit correspondence with the classical quantities.  Section~\ref{sec:implementations} gives the qumode and qubit implementation.  Section~\ref{sec:numerics} reports numerical experiment results.

\section{Classical PCA monitoring background}
\label{sec:classical-background}

For $d\in\mathbb{N}$, let $x_1,\ldots,x_N\in\R^d$ be training vectors assumed to be normal.  Classical PCA monitoring first centers the data,
\begin{equation}
    \bar x\coloneqq \frac1N\sum_{i=1}^N x_i,
    \qquad
    y_i\coloneqq x_i-\bar x,
\end{equation}
and then forms the sample covariance
\begin{equation}
\label{eq:classical-covariance}
    C_{\rm cl}\coloneqq \frac1N\sum_{i=1}^N y_iy_i^\top.
\end{equation}
This centering step is part of the model, not merely preprocessing: the resulting statistics detect deviations around the mean vector.  A test input $x$ is therefore evaluated through its centered form $z=x-\bar x$.

Write the spectral decomposition of the sample covariance as
\begin{equation}
\label{eq:classical-spectrum}
    C_{\rm cl}=\sum_{j=1}^d\lambda_j u_ju_j^\top,
    \qquad
    \lambda_1\ge\lambda_2\ge\cdots\ge\lambda_d\ge0,
\end{equation}
and let
\begin{equation}
    P_K\coloneqq \sum_{j=1}^K u_ju_j^\top
\end{equation}
be the rank-$K$ principal projector.  The rank $K$ is often chosen by a given retained variance target $\alpha \in (0, 1)$:
\begin{equation}
\label{eq:classical-retained}
    K_\alpha\coloneqq 
    \min\!\left\{K:
    \frac{\sum_{j=1}^K\lambda_j}{\sum_{j=1}^d\lambda_j}\ge\alpha
    \right\}.
\end{equation}
For a centered test vector $z$, the classical residual statistic called $Q$ is defined as
\begin{equation}
\label{eq:classical-Q}
    Q_K(z)\coloneqq \left\|(I-P_K)z\right\|^2
    =\sum_{j>K}\left|\langle u_j,z\rangle\right|^2.
\end{equation}
It is large when the input has energy outside the retained principal subspace, which is commonly referred to as the nominal support.

Hotelling's statistic instead measures how the test vector is distributed inside the retained principal subspace.
It uses the squared components of $z$ along the retained eigenvectors $u_1,\ldots,u_K$, and weights each component by the inverse of the corresponding covariance eigenvalue:
\begin{equation}
\label{eq:classical-T2}
    \calT^{2}_{K,\gamma}(z)
    \coloneqq \sum_{j=1}^{K}\frac{|\langle u_j,z\rangle|^2}{\lambda_j+\gamma},
    \qquad \gamma\ge0.
\end{equation}
The ridge parameter $\gamma$ is included to make $\calT^{2}$ stable when retained eigenvalues are small or poorly estimated.

The two statistics have different roles.  In one-class support monitoring, $Q$ is usually the default score because it directly tests whether a sample leaves the principal subspace and does not require inverse spectral weights.  
Hotelling's $\calT^{2}$ is useful when anomalies are expected to remain within the nominal support but have abnormal leverage, for example, a large score in a low-variance retained eigenvector. It requires more structure: a reliable choice of which eigenvectors are retained, stable variance denominators, and a calibration rule for the resulting normalized score. QSPADE maintains this distinction.  The quantum residual score is defined in Section~\ref{subsec:qsad-Q}; the quantum Hotelling-type score is defined in Section~\ref{subsec:qsad-T2}.

\section{Nominal support for quantum data}

\label{sec:nominal-support}

We now translate the PCA anomaly detection idea to the quantum setting. 
In classical PCA, the nominal support is the high-variance spectral region of the covariance matrix.
For quantum data, the analogous object is the high-eigenvalue spectral region of the average state of the normal dataset $\{\rho_i\}_{i=1}^N$:
\begin{equation}
\label{eq:nominal-mixture}
    C\coloneqq \bar\rho=\frac1N\sum_{i=1}^N\rho_i.
\end{equation}

Let $\mathcal H$ be a $d$-dimensional Hilbert space, and define the effect set
\begin{equation}
    \MM(\mathcal H)\coloneqq \{M\in\Herm(\mathcal H):0\preceq M\preceq I\}.
\end{equation}
For every measurement effect $M$,
\begin{equation}
\label{eq:average-acceptance}
    \Tr(M\bar\rho)=\frac1N\sum_{i=1}^N\Tr(M\rho_i).
\end{equation}
The operational meaning of~\eqref{eq:average-acceptance} is simple: an effect~$M$ accepts normal data with probability $\Tr(MC)$.
Thus, learning a nominal support can be phrased as choosing a measurement effect that accepts the average normal state as much as possible under a rank constraint.

\begin{remark}
\label{rem:centering-convention}
Quantum-native states are already normalized objects: pure states are on the unit sphere up to phase, and density operators have trace one and are positive semidefinite.  Subtracting a mean state generally produces a matrix that is not a state, and subtracting vector amplitudes is not invariant under global phase choices.  Therefore, no default centering operation preserves the physical sample space and the Born-probability interpretation in~\eqref{eq:average-acceptance}.  QSPADE consequently uses the uncentered average state in~\eqref{eq:nominal-mixture}.

For classical data, the recommended convention is different.  If the goal is to reproduce centered PCA monitoring, one centers in the classical domain first, $y_i=x_i-\bar x$ and $z=x-\bar x$, and then embeds the centered vectors into quantum states.  After this embedding, the resulting objects are treated as quantum data by QSPADE.  
\end{remark}

Let a spectral decomposition of $C$ be as follows:
\begin{equation}
\label{eq:C-spectrum}
    C=\sum_{j=1}^d\lambda_j\ketbra{u_j},
    \qquad
    \lambda_1\ge\lambda_2\ge\cdots\ge\lambda_d\ge0.
\end{equation}
The high-eigenvalue eigenspaces of $C$ represent the spectral components most strongly supported by the normal data.
The following proposition formalizes this statement.

\begin{proposition}[Hard spectral detector]
\label{prop:hard-support}
For $1\le K\le d$,
\begin{equation}
\label{eq:hard-support-optimization}
    \max_{\substack{M\in\MM(\mathcal H)\\ \rank(M)\le K}}
    \Tr(MC)
    =\sum_{j=1}^K\lambda_j.
\end{equation}
The maximum is attained by any top-$K$ spectral projector
\begin{equation}
    P_K=\sum_{j=1}^K\ketbra{u_j}.
\end{equation}
If $\lambda_K>\lambda_{K+1}$, the top-$K$ projector is unique.
\end{proposition}
\begin{proof}
See Appendix~\ref{app:proof-hard-support}.
\end{proof}
Since $C$ is a density operator, $\sum_{j=1}^d\lambda_j=1$.
Thus, the quantum analogue of a retained-variance target is a retained-mass target $\alpha\in(0,1]$, for which the hard rank is chosen as
\begin{equation}
K_\alpha
\coloneqq
\min\!\left\{
K:\sum_{j=1}^K\lambda_j\ge \alpha
\right\}.
\end{equation}
This hard choice accepts the first $K_\alpha$ eigenspaces completely and rejects the remaining eigenspaces completely.

Proposition~\ref{prop:hard-support} identifies the standard principal subspace with the optimal hard detector.  The next section replaces this discontinuous hard acceptor by a smooth, optimization-derived measurement effect.

\section{Soft Spectral Detector}
\label{sec:regularized-detector}
The hard projector $P_K$ makes a binary decision on the eigenspaces of $C$: each eigenspace is either fully accepted or fully rejected. It is sensitive to eigenvalue crossings and becomes unstable when the spectral gap near the cutoff is small.
A well-behaved detector should remain a valid effect, increase monotonically with the eigenvalues of $C$, and admit a principled optimization formulation. We derive such a detector by introducing a regularized optimization problem.

We first replace the step response in the binary decision with a smooth monotone response function.
\begin{definition}[Optimization-admissible response]
\label{def:admissible}
A scalar response $f\colon\R\to(0,1)$ is \emph{optimization-admissible} if:
\begin{enumerate}[label=\roman*.,leftmargin=2em]
    \item $f\in C^1(\R)$ and $f'(z)>0$ for every $z\in\R$;
    \item $\lim_{z\to-\infty}f(z)=0$ and $\lim_{z\to+\infty}f(z)=1$;
    \item the quantile function $f^{-1}$ belongs to $L^1(0,1)$;
    \item $\int_0^1f^{-1}(m)\,dm=0$.
\end{enumerate}
The associated generator is
\begin{equation}
\label{eq:phi-f}
    \phi_f(m)\coloneqq \int_0^m f^{-1}(u)\,du,
    \qquad m\in(0,1),
\end{equation}
with continuous extension to $[0,1]$.
\end{definition}

The last condition only fixes the location convention for the response.  Logistic, Gaussian, and other smooth thresholds can be placed in this form after centering.

Fix $T>0$ and $\mu\in\R$.  Define the regularized spectral detector as the optimizer of
\begin{equation}
\label{eq:regularized-objective}
    \min_{0\preceq M\preceq I}
    \left\{
    \mu\Tr(M)-\Tr(MC)+T\Tr\phi_f(M)
    \right\}.
\end{equation}
Here $\mu$ plays the role of a spectral threshold, or equivalently, a Lagrange parameter controlling the detector size, while the temperature parameter $T$ controls the smoothness of the transition. 
The term $-\Tr(MC)$ rewards normal acceptance, $\mu\Tr(M)$ penalizes detector size, and $T\Tr\phi_f(M)$ smooths the occupation pattern.

\begin{theorem}[Soft Spectral Detector]
\label{thm:rsd}
Let $C\succeq0$, $T>0$, and let $f$ be optimization-admissible.  Problem~\eqref{eq:regularized-objective} has the unique solution
\begin{equation}
\label{eq:RSD-filter}
    M^f_{\mu,T}=f\!\left(\frac{C-\mu I}{T}\right).
\end{equation}
Equivalently,
\begin{equation}
\label{eq:RSD-spectrum}
    M^f_{\mu,T}
    =\sum_{j=1}^d
    f\!\left(\frac{\lambda_j-\mu}{T}\right)\ketbra{u_j}.
\end{equation}
\end{theorem}
\begin{proof}
See Appendix~\ref{app:proof-rsd}.
\end{proof}

Let $\Pi(C>\mu I)$ denote the projection onto the eigenspace of $C-\mu I$ with strictly positive eigenvalues, and let $\Pi(C=\mu I)$ denote the projection onto the zero eigenspace of $C-\mu I$.

\begin{proposition}[Sharp-resolution limit]
\label{prop:sharp-resolution}
For fixed $\mu$,
\begin{equation}
\label{eq:sharp-limit}
    \lim_{T\to0^+}M^f_{\mu,T}
    =\Pi(C>\mu I)+f(0)\Pi(C=\mu I).
\end{equation}
In particular, if $\lambda_K>\mu>\lambda_{K+1}$, then $M^f_{\mu,T}\to P_K$.
\end{proposition}
\begin{proof}
See Appendix~\ref{app:proof-sharp-resolution}.
\end{proof}

This parameterized family assigns smooth acceptance weights to the eigencomponents of $C$ at finite resolution, while recovering the hard spectral cutoff in the sharp limit.

\section{QSPADE statistics and calibration targets}

\label{sec:qsad-statistics}

We now turn the soft spectral detector into anomaly scores. 
We define the two QSPADE scores and show how the scalar threshold $\mu$ is empirically calibrated to match a desired retained-mass target (the continuous analogue), and how test states are scored based on this threshold.
The residual $Q$ statistic only needs a retained-mass threshold, calibrated directly with the average state~$C$.  A Hotelling-type statistic additionally needs spectral windows ranks, calibrated with the maximally mixed probe~$\tau_{\rm mm}=I/d$.

\subsection{Residual statistic $Q$ and retained-mass threshold}
\label{subsec:qsad-Q}

For fixed $f$ and $T$, define the retained-mass function
\begin{equation}
\label{eq:retained-variance-function}
    R^f_{C,T}(\mu)
    \coloneqq
    \Tr\!\left(M^f_{\mu,T} C\right)
    =
    \sum_{j=1}^d
    \lambda_j
    f\!\left(\frac{\lambda_j-\mu}{T}\right).
\end{equation}
Since $C=\bar\rho$ is a density operator and $f$ is strictly increasing with
limits $0$ and $1$, $R^f_{C,T}(\mu)$ decreases continuously from $1$ to $0$
as $\mu$ increases from $-\infty$ to $+\infty$.  Thus, for each retained-mass
target $\alpha\in(0,1)$, we choose the unique threshold $\mu_\alpha$ by
\begin{equation}
\label{eq:retained-target-equation}
    R^f_{C,T}(\mu_\alpha)=\alpha.
\end{equation}
The calibrated detector is defined as
\begin{equation}
\label{eq:calibrated-detector}
    M^f_\alpha\coloneqq M^f_{\mu_\alpha,T}.
\end{equation}
For a normalized test state $\sigma$, define the acceptance score and residual anomaly score under the retained detector
\begin{equation}
\label{eq:retained-score}
    S^f_\alpha(\sigma)\coloneqq \Tr(M^f_\alpha\sigma),
    \qquad
    Q^f_\alpha(\sigma)\coloneqq 1-S^f_\alpha(\sigma).
\end{equation}

The anomaly decision cutoff for $Q^f_\alpha$ is fixed in advance. Alternatively, to guarantee a specific false-alarm rate, the decision threshold can be empirically bounded using held-out scores, consistent with standard classical process monitoring.

\subsection{Hotelling's $\calT^{2}$ from soft spectral windows}
\label{subsec:qsad-T2}
A Hotelling-type score requires more information than the residual score: it must resolve how a test state is distributed within the part of the spectrum accepted as normal.
For this purpose, we introduce a finite sequence of thresholds and take differences between consecutive soft detectors.
These differences define soft spectral windows.
Each window collects contributions from eigencomponents lying mainly between two neighboring thresholds, with smooth boundary weights determined by the resolution $T$.
The number of windows, $L$, is chosen so that the retained windows cover the target retained mass.
Choose ordered thresholds
\begin{equation}
\label{eq:threshold-sequence}
    +\infty=\mu_0>\mu_1>\cdots>\mu_L,
\end{equation}
and set
\begin{equation}
    M_0\coloneqq0,
    \qquad
    M_\ell\coloneqq M^f_{\mu_\ell,T}
    \quad (\ell=1,\ldots,L).
\end{equation}
The window effects are defined as
\begin{equation}
\label{eq:window-effects}
    D_\ell\coloneqq M_\ell-M_{\ell-1}.
\end{equation}
They are positive semidefinite because lowering the threshold increases the spectral response.  If the final threshold is $-\infty$, then the windows sum to $I$.

For component-wise Hotelling behavior, these retained window boundaries should be rank-calibrated:
\begin{equation}
\label{eq:rank-window-calibration}
    \Tr(M^f_{\mu_\ell,T})=\ell,
    \qquad
    \ell=1,\ldots,L \quad (L \ll d).
\end{equation}
This condition places the threshold \(\mu_\ell\) at the \(\ell\)-th PCA rank position in the soft spectrum. Hence the window
\(D_\ell\)
is the soft analogue of the \(\ell\)-th principal component. 

For each window, define the retained window mass
\begin{equation}
\label{eq:window-denominators}
    v_\ell\coloneqq \Tr(D_\ell C),
\end{equation}
Since \(D_\ell\) is calibrated to represent one spectral component, \(v_\ell\) plays the role of the corresponding soft eigenvalue. 
The total number of retained windows $L$ can be chosen as the minimal integer $L$ satisfying
\begin{equation}
    \sum_{\ell=1}^L v_\ell\ge\alpha.
\end{equation}
For a test state $\sigma$, the test window score is
\begin{equation}
\label{eq:test-window-mass}
    p_\sigma(\ell)\coloneqq \Tr(D_\ell\sigma).
\end{equation}
For a ridge $\gamma\ge0$, define the soft leverage score
\begin{equation}
\label{eq:qsad-T2}
    \mathcal{T}^{2,f}_{L,\gamma}(\sigma)
   \coloneqq \sum_{\ell=1}^L
    \frac{p_\sigma(\ell)}{v_\ell+\gamma}.
\end{equation}

The key calibration distinction is that \(Q\) requires only one retained-mass threshold $\mu_{\alpha}$, whereas \(\mathcal T^2\) requires component-resolved windows in addition, corresponding to thresholds $\mu_\ell$ for $L$ windows. 

\begin{table*}[t]
\centering
\begin{tabular}{p{0.07\textwidth}p{0.08\textwidth}p{0.23\textwidth}p{0.15\textwidth}p{0.22\textwidth}p{0.18\textwidth}}
\toprule
Statistic & Default? & Calibrated Quantity & Required Probe & Recommended Method & Sample Complexity \\
\midrule
$Q$ & yes & retained-mass threshold $\mu_\alpha$ & $C$ & qumode (sweep) /\newline qubit (single-target) & $O(\eta^{-2}\log \delta^{-1})$ (\eqref{complexity1},~\eqref{eq:general-qubit-samples}) 
\\
$\mathcal{T}^{2}$ & optional & thresholds $\mu_\ell$ for $L$ windows & $C$, $\tau_{\mathrm{mm}}=I/d$ & qumode & $\widetilde{O}(L d \eta^{-2}\log \delta^{-1})$ for $L \ll d$ (\eqref{complexity2})\\
\bottomrule
\end{tabular}
\caption{Summary of QSPADE anomaly scores, required calibration, and sample complexities. The residual score $Q$ is the default one-class support test, whereas the Hotelling-type score $\mathcal{T}^{2}$ requires additional windows. Here $C$ is the average state, $\tau_{\mathrm{mm}}=I/d$ is the maximally mixed probe (with system dimension $d$), $\mu_\ell$ are thresholds for the $L$ windows, $\eta$ is the target additive accuracy, and $\delta$ is the failure probability.}
\label{tab:implementation-map}
\end{table*}

\begin{remark}[Boundary leverage in the soft residual]
In standard hard PCA, the $Q$-statistic is entirely insensitive to variance differences among the retained components, making $\mathcal{T}^2$ strictly necessary for leverage monitoring. In QSPADE, however, the soft residual $Q^f_\alpha$ applies a graded penalty $1-f((\lambda_j-\mu)/T)$ to the boundary modes. Because this fractional penalty naturally tracks the spectral decay in the transition region, the soft $Q$ statistic absorbs some of the variance-weighting functions of a traditional leverage score near the cutoff. Consequently, for many cases, the $Q$ statistic alone is sufficiently informative, relegating $\mathcal{T}^2$ to cases where anomalous leverage deep within the dominant support is explicitly suspected.
\end{remark}

\begin{proposition}[Sharp-limit correspondence with classical $Q$ and $\mathcal{T}^{2}$]
\label{prop:sharp-QT2}
Suppose that $C$ has nondegenerate positive eigenvalues in the retained range. We consider the limit $T\to0^+$, with thresholds placed between consecutive eigenvalues. Then, in this limit,
\begin{equation}
    D_j\to\ketbra{u_j}.
\end{equation}
For a pure test state $\sigma=\ketbra{\psi}$ with $\ket{\psi}=\sum_j a_j\ket{u_j}$, the residual statistic converges to
\begin{equation}
    Q^f_\alpha(\sigma)\to\sum_{j>L}|a_j|^2,
\end{equation}
and the soft leverage score converges to
\begin{equation}
    \mathcal{T}^{2,f}_{L,\gamma}(\sigma)
    \to
    \sum_{j=1}^L\frac{|a_j|^2}{\lambda_j+\gamma}.
\end{equation}
Thus, when $C$ is obtained from centered classical features, QSPADE recovers the classical PCA residual and Hotelling forms. 
\end{proposition}
\begin{proof}
See Appendix~\ref{app:proof-sharp-QT2}.
\end{proof}

\section{Quantum measurement implementations}
\label{sec:implementations}
The soft spectral detector $M^f_{\mu,T}$ can be accessed without explicit eigenvector reconstruction. Table~\ref{tab:implementation-map} summarizes the two calibration objects. The qumode implementation records a continuous spectral sample, making it naturally suited for global retained-mass sweeps and window-resolved $\mathcal{T}^{2}$ leverage. Conversely, the qubit implementation estimates a single target threshold via Hadamard tests for the residual $Q$ detector.

Throughout this section, we use \(C=\bar\rho\) as a Hamiltonian generator. 
This access can be supplied either by a Hamiltonian-simulation oracle or block-encoding for \(C\), or, in the sample-access setting, by density matrix exponentiation~\cite{Lloyd2014} (see also \cite{kimmel2017,go2025}) using copies of \(C\), which are prepared by sampling normal training states uniformly. 
Thus the qumode interaction \(e^{-i\hat p\otimes C/T_2}\) in~\eqref{eq:cv-evolution} and the qubit evolutions \(e^{it_\omega C/T}\) in~\eqref{eq:controlled-U-omega} are both implementations of the same underlying simulation primitive for \(C\). 

\subsection{Qumode spectral-threshold measurement}
\label{subsec:qumode-threshold}

We first present a quantum implementation of the soft detector $M^f_{\mu,T}$ via an ancillary qumode interacting with~$C$, followed by a classical threshold on the position outcome.

Since $f'>0$ and $f(+\infty)-f(-\infty)=1$, the derivative~$f'$ is a probability density function.  For $T_1>0$, define the reflected position density
\begin{equation}
\label{eq:reflected-density}
    g_{f,T_1}(q)\coloneqq \frac1{T_1}f'\!\left(-\frac{q}{T_1}\right).
\end{equation}
Prepare a control qumode in the position wavepacket
\begin{equation}
\label{eq:position-wavepacket}
    \ket{\psi_{f,T_1}}
    \coloneqq \int_{\R}\sqrt{g_{f,T_1}(q)}\ket q\,dq.
\end{equation}
The response \(f\) is chosen from the optimization-admissible class in
Definition~\ref{def:admissible}.  This choice is
not only a regularization choice: it also fixes the physical control
distribution through \(f'\).  Two useful examples are the Fermi--Dirac (logistic) response, which connects to entropy-regularized soft spectral measurements~\cite{LiuWilde2026,Yuan2026QPCA}, and
the Gaussian (probit) response, whose control qumode is a
standard Gaussian state~\cite{Weedbrook2012Gaussian} and connects to the Gaussian error function activation observable from \cite{HeLiuWilde2026FDM}.  Their explicit densities and details are given in
Appendix~\ref{app:response-examples}.

Apply
\begin{equation}
\label{eq:cv-evolution}
    U_C\coloneqq e^{-i\hat p\otimes C/T_2},
    \qquad T_2>0,
\end{equation}
measure the position quadrature $\hat{q}$, and accept if the outcome $q>\beta$, where $\beta\in \mathbb{R}$ is chosen based on $\mu$ and $T$.

\begin{proposition}[General spectral-threshold measurement]
\label{prop:cv-threshold}
The above procedure implements the binary POVM effect
\begin{equation}
\label{eq:cv-effect}
    E^f_{\beta,T_1,T_2}
    \coloneqq f\!\left(\frac{C-\beta T_2 I}{T_1T_2}\right).
\end{equation}
Equivalently, with
\begin{equation}
    T\coloneqq T_1T_2,
    \qquad
    \mu\coloneqq \beta T_2,
\end{equation}
one has $E^f_{\beta,T_1,T_2}=M^f_{\mu,T}$.
\end{proposition}
\begin{proof}
See Appendix~\ref{app:proof-cv-threshold}.
\end{proof}

The shift $\mu I$ is not simulated in the circuit. Once $U_C$ and the control width are fixed, changing $\mu$ only changes the classical threshold $\beta$ applied to the measured position.

\subsection{Qumode calibration for $Q$ and $\calT^{2}$}
\label{subsec:qumode-calibration}

We now explain how the qumode-based detector introduced in Section~\ref{subsec:qumode-threshold} is calibrated for $Q$ and $\calT^{2}$.

For the residual statistic $Q$, prepare the average state~$C$.  With the circuit in Proposition~\ref{prop:cv-threshold}, define
\begin{equation}
\label{eq:cv-tail-retained}
    G_C(\beta)\coloneqq \Pr(q>\beta\mid C).
\end{equation}
Then
\begin{equation}
\label{eq:retained-born}
    G_C(\beta)
    =\Tr(E^f_{\beta,T_1,T_2} C)
    =\Tr(M^f_{\beta T_2,T_1T_2} C).
\end{equation}
Thus a retained-mass target $\alpha$ is calibrated by
\begin{equation}
    G_C(\beta_\alpha)=\alpha,
    \qquad
    \mu_\alpha=\beta_\alpha T_2.
\end{equation}
Because the threshold $\beta$ is applied in post-processing, the same recorded position samples can be reused to calibrate any number of targets $\alpha$. This global sweep procedure is summarized in Algorithm~\ref{alg:all-target-retained}.

\begin{algorithm}[H]
\caption{Qumode retained-mass calibration for residual $Q$}
\label{alg:all-target-retained}
\begin{algorithmic}[1]
\Require Average state $C$, response $f$, control scales $T_1,T_2$, successful calibration shots $S$, target set $\mathcal A\subset(0,1)$.
\State Prepare $\ket{\psi_{f,T_1}}$ with density~\eqref{eq:reflected-density}.
\For{$s=1,\ldots,S$}
    \State Prepare $C$, apply $U_C=e^{-i\hat p\otimes C/T_2}$, and record the position outcome $q_s$.
\EndFor
\State Form the empirical tail function
\[
    \widehat G_S(\beta)=\frac1S\sum_{s=1}^S\mathbf{1}\{q_s>\beta\}.
\]
\For{each $\alpha\in\mathcal A$}
    \State Choose $\widehat\beta_\alpha$ satisfying $\widehat G_S(\widehat\beta_\alpha)\approx\alpha$ and set $\widehat\mu_\alpha=\widehat\beta_\alpha T_2$.
\EndFor
\State \Return thresholds $\{\widehat\mu_\alpha:\alpha\in\mathcal A\}$.
\end{algorithmic}
\end{algorithm}

Uniformly over $\beta$, the Dvoretzky--Kiefer--Wolfowitz inequality~\cite{Dvoretzky1956} gives
\begin{equation}
\label{eq:retained-DKW-samples}
    \Pr\left(\sup_\beta|\widehat G_S(\beta)-G_C(\beta)|\le\eta\right)\ge1-\delta
\end{equation}
if
\begin{equation}
    \label{complexity1}
    S\ge\frac1{2\eta^2}\log\frac2\delta .
\end{equation}
Thus, the retained-mass curve is learned to additive accuracy $\eta$ with $O(\eta^{-2}\log\delta^{-1})$ samples, independent of $d$.

For $\calT^{2}$, the same qumode measurement circuit is used to construct component-resolved windows $D_\ell$.
The maximally mixed probe
\(
    \tau_{\rm mm}\coloneqq \frac{I}{d}
\)
gives
\begin{equation}
\label{eq:rank-calibration-born}
\begin{aligned}
    \Pr(q>\beta\mid\tau_{\rm mm})
    & =\Tr(E^f_{\beta,T_1,T_2}\tau_{\rm mm})\\
    & =\frac1d\Tr(M^f_{\beta T_2,T_1T_2}).
    \end{aligned}
\end{equation}
Thus, the rank-calibrated thresholds are obtained from
\begin{equation}
\label{eq:rank-quantile-equation}
    \Pr(q>\beta_\ell\mid\tau_{\rm mm})=\frac{\ell}{d}.
\end{equation}
For windows $W_\ell=(\beta_\ell,\beta_{\ell-1}]$, the Hotelling denominator is then the retained window mass
\begin{equation}
\label{eq:window-denominator-estimates}
    v_\ell=\Pr(q\in W_\ell\mid C)
    =\Tr(D_\ell C).
\end{equation}

Algorithm~\ref{alg:T2-score} below gives the full Hotelling-type scoring procedure: the maximally mixed probe fixes component-resolved windows, and the average state \(C\) supplies the denominators.

\begin{algorithm}[H]
\caption{Qumode Hotelling-type score}
\label{alg:T2-score}
\begin{algorithmic}[1]
\Require Same fixed qumode circuit as Algorithm~\ref{alg:all-target-retained}; average state \(C\); test state \(\sigma\); maximally mixed state \(\tau_{\mathrm{mm}}=I/d\); retained-mass target \(\alpha\); candidate window number \(L_{\max}\); ridge \(\gamma\ge0\).
\State Use \(\tau_{\mathrm{mm}}\) to choose thresholds \(\widehat\beta_\ell\) satisfying
\[
    \widehat{\Pr}(q>\widehat\beta_\ell\mid\tau_{\rm mm})
    \approx \frac{\ell}{d},
    \qquad \ell=1,\ldots,L_{\max}.
\]
\State Define windows \(W_\ell=(\widehat\beta_\ell,\widehat\beta_{\ell-1}]\), with \(\widehat\beta_0=+\infty\).
\State Estimate retained window masses
\[
    \widehat v_\ell=\widehat{\Pr}(q\in W_\ell\mid C),
    \qquad \ell=1,\ldots,L_{\max}.
\]
\State Choose
\[
    \widehat L=\min\!\left\{L':\sum_{\ell=1}^{L'}\widehat v_\ell\ge\alpha\right\}.
\]
\State Estimate test-window masses
\[
    \widehat p_\sigma(\ell)=\widehat{\Pr}(q\in W_\ell\mid\sigma),
    \qquad \ell=1,\ldots,\widehat L.
\]
\State Output
\[
    \widehat{\mathcal T}^{2,f}_{\widehat L,\gamma}(\sigma)
    =
    \sum_{\ell=1}^{\widehat L}
    \frac{\widehat p_\sigma(\ell)}
         {\widehat v_\ell+\gamma}.
\]
\end{algorithmic}
\end{algorithm}

To calibrate the thresholds $\mu_\ell$ for the retained windows to an additive accuracy $\eta$ for the rank level, the required probability accuracy in the $\tau_{\mathrm{mm}}$ probe is $\eta/d$. Because we restrict Hotelling leverage monitoring to $L \ll d$ principal windows, the relevant binomial success probability for these top modes under the maximally mixed probe is small, $p \approx L/d$. Applying Bernstein's inequality~\cite{Boucheron2013} leverages the reduced variance $p(1-p) \approx L/d$, yielding a significantly sharper sample complexity:
\begin{equation}
    \label{complexity2}
    S_{\calT^{2}}
    =
    O\!\left(
    \left(
    \frac{L d}{\eta^2}
    +
    \frac{d}{\eta}
    \right)
    \log\frac{1}{\delta}
    \right).
\end{equation}
Thus, by exploiting the finite-window truncation $L \ll d$, QSPADE breaks the naive $\widetilde{O}(d^2)$ worst-case scaling, reducing the component-resolution rank calibration overhead to $\widetilde{O}(Ld)$. 

\subsection{Single-target qubit calibration for residual $Q$}
\label{subsec:qubit-Q}

When the retained-mass target $\alpha$ is fixed in advance, it is unnecessary to collect a full qumode distribution. A qubit route can calibrate the single scalar threshold~$\mu_\alpha$ by directly evaluating the nonlinear spectral filter \( \Tr[f((C-\mu I)/T) \rho ] \) using a hybrid Hadamard-test estimator adapted from Algorithm~2 of~\cite{HeLiuWilde2026FDM}. 

Since a standard Hadamard test~\cite{Cleve1998} evaluates expectation values of the form \( \Tr(e^{iHt} \rho ) \), we bridge the gap to the nonlinear response $f$ by representing $f(x)$ as a bounded superposition of phase functions \(e^{itx}\) on the relevant spectral interval (Definition~\ref{def:hadamard-estimable-response}). By substituting the operator argument into this phase expansion, evaluating the target filter decomposes into estimating a sequence of unitary terms $e^{it(C-\mu I)/T}$.

Because the scalar shift $\mu I$  commutes with $C$,  the evolution $e^{it(C-\mu I)/T}$ trivially factors as
\begin{equation}
\label{eq:qubit-control-shift-identity}
    e^{it(C-\mu I)/T}
    =
    e^{-i\mu t/T}e^{itC/T}.
\end{equation}
Thus the data register only needs controlled simulations of $C$; changing the threshold $\mu$ only changes a scalar phase on the control qubit.

\begin{definition}[Hadamard-estimable response]
\label{def:hadamard-estimable-response}
Fix a compact interval $\calJ\subset\R$.  A response $f\colon\R\to(0,1)$ is called $(A_f,\varepsilon_f)$-Hadamard-estimable on $\calJ$ if there exist a real constant $b_f$, a constant $A_f>0$, a probability space~$\Omega_f$, real frequencies $t_\omega\in\R$, and coefficients $\chi_\omega\in\C$ with $|\chi_\omega|\le1$ such that
\begin{equation}
\label{eq:hadamard-representation}
    \widetilde f(x)
    \coloneqq 
    b_f+
    A_f\,\EE_{\omega\sim\Omega_f}
    \Re\!\left[\chi_\omega e^{it_\omega x}\right]
\end{equation}
satisfies
\begin{equation}
\label{eq:hadamard-representation-error}
    \sup_{x\in\calJ}|f(x)-\widetilde f(x)|\le \varepsilon_f.
\end{equation}
\end{definition}

The interval $\calJ$ is chosen so that $(\lambda-\mu)/T\in\calJ$ for all $\lambda\in\spec(C)$ and all thresholds $\mu$ considered during calibration.  Since $C=\bar\rho$ is a density operator, $\spec(C)\subseteq[0,1]$; hence for a calibration interval $[\mu_{\min},\mu_{\max}]$ one may take
\begin{equation}
\label{eq:qubit-safe-J}
    \calJ
    \coloneqq 
    \left[
    -\frac{\mu_{\max}}{T},
    \frac{1-\mu_{\min}}{T}
    \right],
\end{equation}
or a smaller interval if tighter spectral bounds are known.   
By casting $f$ into this phase-superposition form, the evaluation of the nonlinear filter translates directly into random sampling over controlled-unitary evolutions. 
Representations of the form in~\eqref{eq:hadamard-representation} arise from Fourier approximations of \(f\) or other truncated phase expansions. They naturally accommodate quantum signal processing (QSP) polynomials~\cite{LowChuang2017QSP,Gilyen2019} or linear combinations of unitaries (LCU)~\cite{Berry2015LCU}. The quantity $A_f$ captures the overall sampling prefactor or block-encoding norm of the chosen representation.

To construct the single-shot estimator $Y_\mu(\rho)$ for the target filter \( \Tr[f((C-\mu I)/T) \rho ] \), the process is as follows:

\begin{enumerate}
    \item Sample $\omega\sim\Omega_f$ and write $\chi_\omega=|\chi_\omega|e^{i\varphi_\omega}$ (choosing an arbitrary phase $\varphi_\omega$ if $\chi_\omega=0$).
    \item Prepare the data register in state $\rho$ and a control qubit in $\ket{0}$.  Apply a Hadamard gate to the control, apply the controlled unitary
    \begin{equation}
    \label{eq:controlled-U-omega}
        U_\omega\coloneqq e^{it_\omega C/T}
    \end{equation}
    on the data register, apply the phase
    \begin{equation}
    \label{eq:hadamard-control-phase}
        \theta_\omega(\mu)
        \coloneqq 
        \varphi_\omega-\frac{\mu t_\omega}{T}
    \end{equation} 
    to the control qubit.
    \item Measure the control qubit in the $X$ basis to obtain a signed outcome $Z_\omega(\rho)\in\{-1,+1\}$.
\end{enumerate}

Conditioned on the sampled frequency $\omega$, the expected value of this measurement outcome exactly recovers the required real part:
\begin{equation}
\EE[Z_\omega(\rho)\mid\omega] = \Re[e^{i\theta_\omega(\mu)}\Tr(U_\omega\rho)].  
\end{equation}
Scaling and translating this outcome yields the final single-shot estimator:
\begin{equation}
\label{eq:single-shot-Y}
    Y_\mu(\rho)
    \coloneqq 
    b_f+A_f|\chi_\omega|Z_\omega(\rho).
\end{equation}

\begin{proposition}[Control-shifted Hadamard estimator]
\label{prop:control-shifted-estimator}
Suppose that $f$ is $(A_f,\varepsilon_f)$-Hadamard-estimable on $\mathcal{J}$ and that $C$ can be simulated for the times $t_\omega/T$. For every density operator $\rho$ and every $\mu$ satisfying $\left(\spec(C)-\mu\right)/T\subseteq\mathcal{J}$, the single-shot estimator $Y_\mu(\rho)$ defined in~\eqref{eq:single-shot-Y} satisfies
\begin{equation}
\label{eq:general-estimator-bias}
    \left|
    \EE\,Y_\mu(\rho)
    -
    \Tr\!\left[
    f\!\left(\frac{C-\mu I}{T}\right) \rho 
    \right]
    \right|
    \le \varepsilon_f.
\end{equation}
Moreover, the estimator is strictly bounded, $Y_\mu(\rho)\in[b_f-A_f,b_f+A_f]$.
\end{proposition}
\begin{proof}
See Appendix~\ref{app:proof-control-shifted-estimator}.
\end{proof}

For fixed $\mu$ and $\rho$, averaging $m$ independent copies of $Y_\mu(\rho)$ gives, by Hoeffding's inequality~\cite{Hoeffding1963}, additive error at most $\eta+\varepsilon_f$ with probability at least $1-\delta$ using
\begin{equation}
\label{eq:general-qubit-samples}
    m=
    O\!\left(
    \frac{A_f^2}{\eta^2}\log\frac1\delta
    \right)
\end{equation}
Hadamard-test samples, multiplied by the average or worst-case cost of simulating $e^{it_\omega C/T}$ over the sampled frequencies. 

For residual calibration, the probe state is the average state $C=\bar\rho$.  Operationally, this means sampling a normal training state uniformly.  With $\rho=C$, Proposition~\ref{prop:control-shifted-estimator} gives noisy observations of
\begin{equation}
\label{eq:qubit-retained-estimate-target}
    R^f_{C,T}(\mu)
    =
    \Tr\!\left[f\!\left(\frac{C-\mu I}{T}\right)
    C 
    \right].
\end{equation}
The threshold $\mu_\alpha$ is then found by Robbins--Monro stochastic approximation~\cite{RobbinsMonro1951} applied to the monotone equation
\begin{equation}
    R^f_{C,T}(\mu)=\alpha.
\end{equation}

\begin{algorithm}[H]
\caption{Single-target qubit calibration for residual~$Q$}
\label{alg:single-target-rm}
\begin{algorithmic}[1]
\Require Hadamard-estimable response $f$, resolution $T$, target retained mass $\alpha\in(0,1)$, interval $[\mu_{\min},\mu_{\max}]$ containing the solution, initial $\mu_0\in[\mu_{\min},\mu_{\max}]$, step sizes $\eta_n$, average state $C=\bar\rho$, mini-batch sizes $m_n\ge1$, number of iterations $N$.
\For{$n=0,1,\ldots,N-1$}
    \State Draw independent samples
    \[
        Y_{\mu_n}^{(1)}(C),\ldots,Y_{\mu_n}^{(m_n)}(C)
    \]
    using~\eqref{eq:single-shot-Y}.
    \State Form
    \[
        \widehat R_n(\mu_n)
        =
        \frac1{m_n}\sum_{s=1}^{m_n}Y_{\mu_n}^{(s)}(C).
    \]
    \State Update
    \[
        \mu_{n+1}
        =
        \Pi_{[\mu_{\min},\mu_{\max}]}
        \left(
        \mu_n+\eta_n(\widehat R_n(\mu_n)-\alpha)
        \right).
    \]
\EndFor
\State \Return $\widehat\mu_\alpha=\mu_N$.
\end{algorithmic}
\end{algorithm}

Algorithm~\ref{alg:single-target-rm} implements this stochastic calibration. The sign in the update is chosen because $R^f_{C,T}$ is decreasing: if the observed retained mass is larger than~$\alpha$, then the threshold is too low and must be increased. Under bounded conditional variance and the standard Robbins--Monro step-size conditions
\begin{equation}
    \sum_{n=1}^{\infty}\eta_n=\infty,
    \qquad
    \sum_{n=1}^{\infty}\eta_n^2<\infty
\end{equation}
satisfied, for instance, by an algebraically decaying sequence $\eta_n = c/n^\nu$ with $\nu \in (0.5, 1]$, the recursion converges almost surely to $\mu_\alpha$ as $N\to\infty$ when the estimator is unbiased.
A deterministic representation bias $\varepsilon_f$ moves the limiting point to an $O(\varepsilon_f/\kappa)$ neighborhood of $\mu_\alpha$ whenever the local slope satisfies $|(R^f_{C,T})'(\mu)|\ge\kappa>0$ around the solution. In practical quantum implementations subject to finite-shot binomial noise, the raw trajectory $\mu_n$ may exhibit oscillations. Applying Polyak--Ruppert averaging~\cite{Ruppert1988,PolyakJuditsky1992} $\bar{\mu}_N = \frac{1}{N}\sum_{n=1}^N \mu_n$ over the iterations provides a robust estimate of the true threshold.

\subsection{Deployment}
Regardless of the calibration route, once the calibrated threshold
$\widehat{\mu}_\alpha$ is obtained, the residual detector is fixed as $M^f_{\widehat{\mu}_\alpha,T}$. Deployment on a test state $\sigma$
estimates the acceptance score and the residual anomaly score,
\begin{equation}
    S^f_\alpha(\sigma)=\Tr(M^f_{\widehat{\mu}_\alpha,T}\sigma),
    \qquad
    Q^f_\alpha(\sigma)=1-S^f_\alpha(\sigma).
\end{equation}
The same measurement primitive is used with the calibration input replaced by
the test state $\sigma$.  For a binary qumode threshold measurement,
additive score accuracy $\eta$ requires
$O(\eta^{-2}\log\delta^{-1})$ shots.  For the qubit Hadamard route, the
corresponding sample count is
$O(A_f^2\eta^{-2}\log\delta^{-1})$, with the per-shot cost determined by
the simulation of $e^{it_\omega C/T}$ in the chosen representation.

If Hotelling leverage $\mathcal T^2$ is requested, the qumode implementation is the natural route: its continuous position output
$q$ can be binned into all calibrated windows from the same set of measurement records. The deployment on a test state $\sigma$ is included in Algorithm~\ref{alg:T2-score}.    

\section{Numerical experiments}
\label{sec:numerics}
We report two simulations that test the two intended uses of QSPADE.
The first uses classical data presented as quantum feature states and is a
consistency check: the QSPADE residual and leverage scores reproduce kernel-PCA-like nonlinear spectral geometry.
The second is the genuinely quantum-native setting, in which the normal object
is the mixed state $C=\bar\rho=N^{-1}\sum_i\ketbra{\psi_i}$ of a transverse-field Ising model.

\subsection{Classical data: the kernel-PCA behavior}
\label{subsec:expA}

\begin{figure}[t]
\centering
\includegraphics[width=\linewidth]{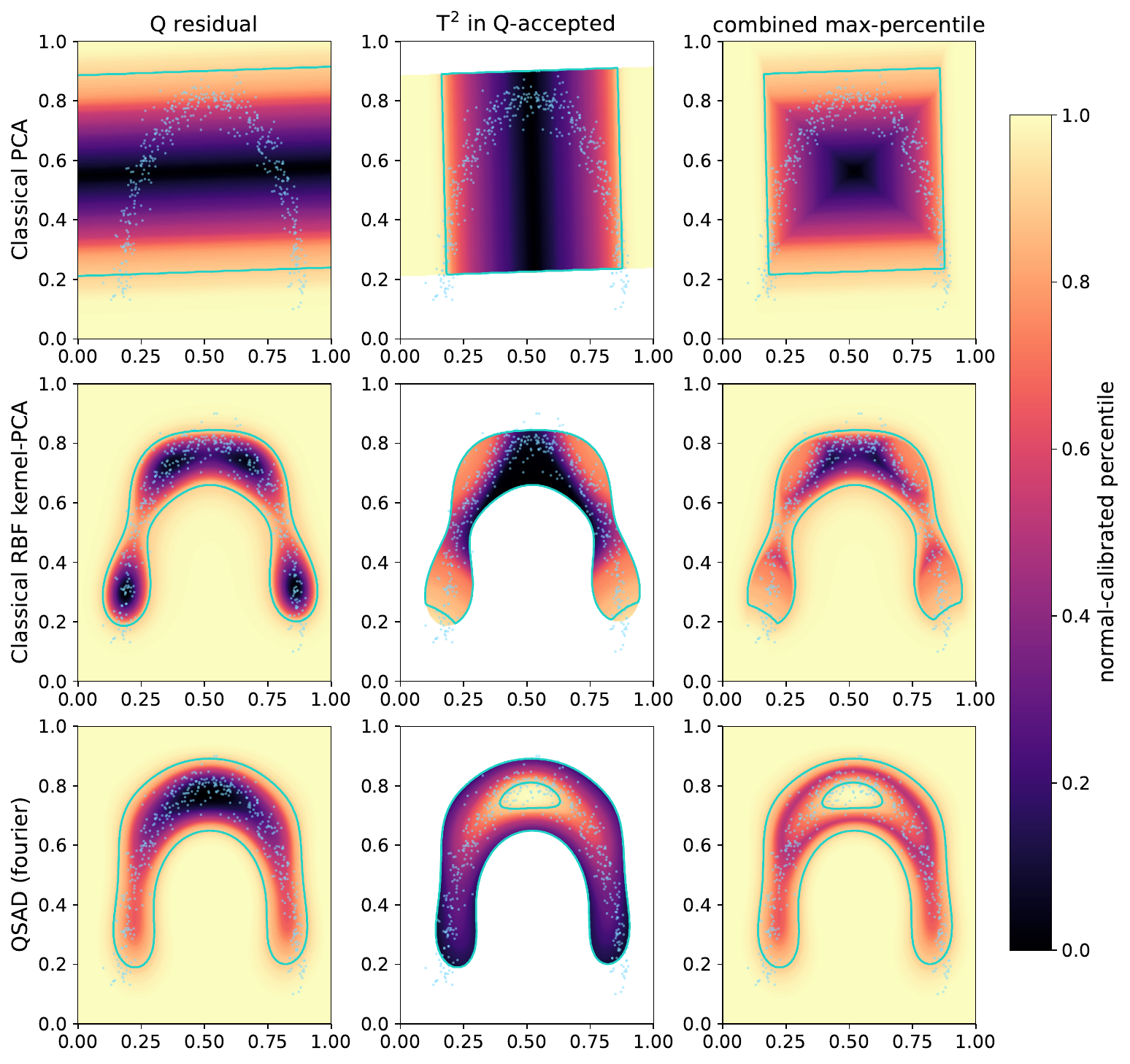}
\caption{Normal-calibrated percentile maps over
$[0,1]^2$ for the moon dataset; the pale scattered points are the normal training
data.  Rows: linear classical PCA, RBF kernel-PCA, and QSPADE with a Fourier feature map.  Columns: residual~$Q$, within-support
$\calT^{2}$ (shown inside the $Q$-accepted region), and the combined
max-percentile.  Linear PCA gives a rigid straight band; the kernel baseline and
QSPADE both produce curved acceptance regions that follow the data.  The cyan
contour marks the calibrated decision boundary.}
\label{fig:expA}
\end{figure}

The normal data in this numerical experiment are a one-class curved point cloud (a single ``moon'') in
$[0,1]^2$.  Each point $x$ is mapped to a normalized feature state
$\ket{\phi(x)}$ by a fixed feature map after classical centering.  
Because the QSPADE scores depend on the data only through the inner products
$\braket{\phi(x)}{\phi(x')}=k(x,x')$, QSPADE on classical data is a kernel method, and in the sharp-resolution limit, the residual $Q^f_\alpha$ coincides
with the kernel-PCA reconstruction error.  This experiment makes the
correspondence explicit and quantitative.

We compare three monitors on identical data: (i)~linear PCA in the raw
two-dimensional space (centered, single retained component); (ii)~standard
centered radial basis function (RBF) kernel-PCA, the canonical kernel one-class baseline, with the
median-heuristic bandwidth; and (iii)~QSPADE with a fixed representative feature
map, a four-qubit tensor-product Fourier map.  The feature
maps, kernels, and all simulation
parameters are specified in Appendix~\ref{app:setup}.  For each monitor, in Figure~\ref{fig:expA} we report
the residual $Q$, the within-support Hotelling $\calT^{2}$, and their combined
percentile, calibrated against the empirical distribution of normal scores. Linear PCA produces rigid, straight acceptance bands that cannot follow the
curvature of the data, so normal points near the extremes of the arc are
flagged as anomalous.  The RBF kernel-PCA and QSPADE both produce curved acceptance
regions that track the manifold; the boundaries are similar though not identical,
since they realize different kernels.

In Table~\ref{tab:expA}, we quantify detection accuracy by the area under the receiver-operating-%
characteristic curve (ROC-AUC), a threshold-independent score equal to $1$ for
perfect normal/anomaly separation and $1/2$ for chance.  Across a suite of
geometries that spans an elongated Gaussian blob (the
linear-correlation case), the moon, and a ring, linear PCA degrades from ROC-AUC $\approx0.89$ on the blob to
$\approx0.57$ on the ring, whereas both the kernel baseline and QSPADE remain in
the $0.82$--$0.93$ range (see also Figure~\ref{fig:expA-datasets} in Appendix~\ref{app:numerics}).
The agreement is
not specific to one kernel: a comparison across further classical kernels
(Laplacian and polynomial) and further feature maps shows the same
data-dependent behavior on both sides (Figure~\ref{fig:expA-kernels} in Appendix~\ref{app:numerics}).

\begin{table}[t]
\centering
\begin{tabular}{lccc}
\toprule
geometry & linear PCA & RBF kernel-PCA & \textbf{QSPADE (Fourier)} \\
\midrule
blob & $0.89$ & $0.93$ & \textbf{0.93} \\
moon & $0.66$ & $0.87$ & \textbf{0.88} \\
ring & $0.57$ & $0.82$ & \textbf{0.88} \\
\bottomrule
\end{tabular}
\caption{ROC-AUC for one-class detection on three classical geometries: an elongated Gaussian blob, the moon, and a ring. 
Linear PCA degrades on curved or non-convex dataset, while the RBF kernel-PCA baseline and QSPADE with a Fourier feature map remain comparable. 
QSPADE supplies a PCA-style one-class anomaly score for quantum-kernel methods: whenever a quantum feature map is advantageous for classical data, the same access model can be used for calibrated anomaly detection.}
\label{tab:expA}
\end{table}

\subsection{Quantum-native data: support across a phase transition}
\label{subsec:expB}

We now apply QSPADE to genuinely quantum-native data. This experiment tests whether a soft retained-mass support detector can monitor a quantum phase change when the normal state has a graded spectrum and no clear PCA rank.
We use the transverse-field Ising Hamiltonian
\begin{equation}
    H_{\rm TFIM}\coloneqq -J\sum_{i=1}^{n-1}Z_iZ_{i+1}-h\sum_{i=1}^{n}X_i,
\end{equation}
on $n=10$ spins with $J=1$, which has a quantum phase transition at the
critical field $h_c=J$ between an ordered (ferromagnetic) phase with long-range
$ZZ$ correlations and a disordered (paramagnetic) phase polarized along the
field.  

The normal source, in our setting, populates the low-energy ladder of the ordered
chain ($h_A=0.4$) with geometrically decreasing probabilities: each sample is
an eigenstate $\ket{m_j}$ of \(H_{\rm TFIM}(h_A)\) drawn with probability $w_j\propto2^{-j}$ over the
$M=12$ lowest modes, with a small random perturbation, so that
$C=N^{-1}\sum_i\ketbra{\psi_i}$, built from $N=600$ such samples, is a genuinely
mixed state with a graded
spectrum.  This thermal-like occupation makes the retained rank ambiguous by
construction: the empirical eigenvalues of $C$ (Fig.~\ref{fig:expB-spectrum})
decay geometrically from $0.52$ through the sampling resolution
$1/N\approx1.7\times10^{-3}$ into the noise floor, the expected per-mode
training counts $w_jN$, falling from $300$ to below one inside the ladder, and no
spectral gap singles out a correct number of components. 

Anomalies are instead
low-energy states of the same chain tuned to the paramagnetic phase at
$h_B=1.2$, drawn from a $0.7/0.2/0.1$
mixture of its three lowest states; we call states from this phase the anomaly
class. 
All reported scalar scores are averages over fresh samples from the corresponding source: the ordered-phase low-energy mixture for normal states, the paramagnetic mixture for the anomaly class, and the TFIM ground state at the displayed field values for the sweep experiment. The evaluation protocol is detailed in Appendix~\ref{app:setup}.

\begin{figure}[t]
\centering
\includegraphics[width=\linewidth]{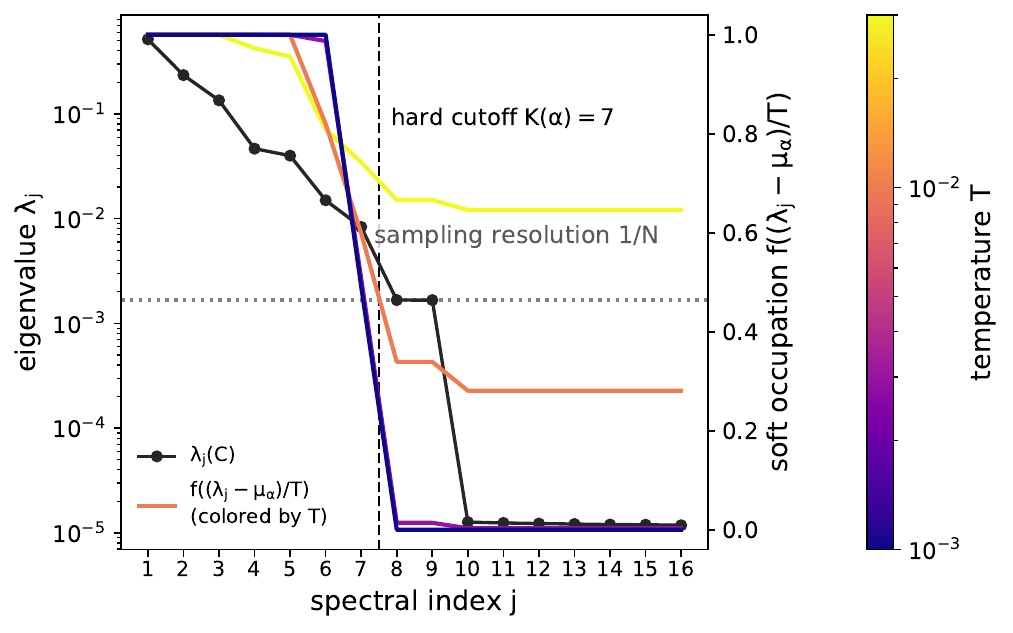}
\caption{Spectrum and soft occupations of the average normal state \(C\).
Black points show the eigenvalues \(\lambda_j(C)\) on a log scale, and the dotted line marks the sampling resolution \(1/N\) for \(N=600\).
The vertical dashed line marks the hard retained-mass rank \(K(\alpha)=7\) for \(\alpha=0.99\).
Colored curves, read on the right axis, show the soft occupations \(f((\lambda_j-\mu_\alpha)/T)\) for different resolutions \(T\).
The figure shows why a hard rank is ambiguous on this graded spectrum, while QSPADE assigns fractional weights to near-cutoff modes.}
\label{fig:expB-spectrum}
\end{figure}

We compare three residual anomaly scores: 
\begin{itemize}
    \item \textbf{Raw overlap residual:} 
    $Q_{\rm raw}(\psi)=1-\bra{\psi}C\ket{\psi}$, which is density-weighted.
    
    \item \textbf{Hard top-$K$ residual:} 
    $Q_{\rm hard}(\psi)=1-\bra{\psi}P_K\ket{\psi}$, with the rank $K$ to be chosen.
    
    \item \textbf{QSPADE residual:} 
    $Q^f_\alpha(\psi)=1-\bra{\psi}M^f_{\mu_\alpha,T}\ket{\psi}$, fixed by the single retained-mass target $\alpha=0.99$.
\end{itemize}
On a graded spectrum the rank choice that
hard PCA requires is ill-posed, in the sense that the data single out no
particular $K$: the retained-mass rule maps the reasonable
targets $\alpha\in\{0.95,0.98,0.99,0.995\}$ to four different ranks
$K=5,6,7,8$, and beyond $K=7$ the cutoff reaches the sampling resolution,
where the trailing empirical eigenvectors are no longer reliable.  The
retained-mass calibration instead places the threshold $\mu_\alpha$ inside the
graded tail, between $\lambda_7$ and $\lambda_8$, and the soft response retains
the boundary modes fractionally, a weighting that no projector can represent
(Fig.~\ref{fig:expB-spectrum}).  The resolution is chosen to resolve the
spectral scale of the tail: the scalar scores quoted below use
$T=3\times10^{-3}$, and the figures display a family of resolutions around it.

\begin{figure}[t]
\centering
\includegraphics[width=\linewidth]{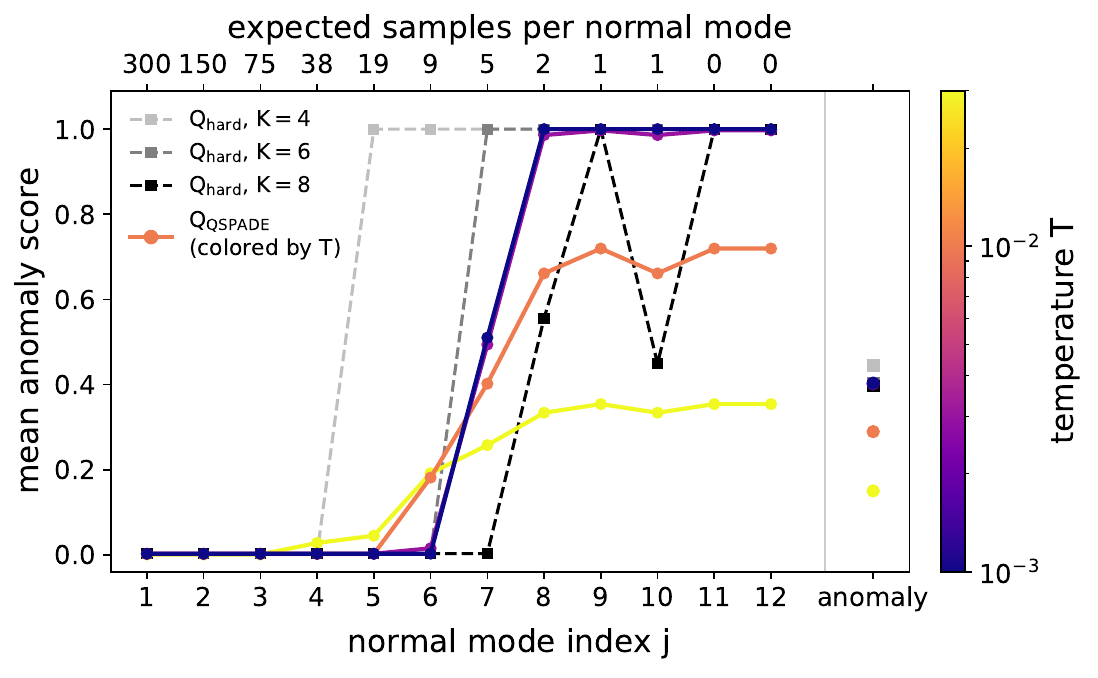}
\caption{Mean residual anomaly score for each normal mode \(j\), with the expected number of training samples \(w_jN\) shown on the top axis, and for the anomaly class in the rightmost column.
Grey dashed curves are hard top-\(K\) projectors; colored curves are QSPADE residuals at different resolutions \(T\).
Hard ranks flip whole normal modes between zero and unit score, and become unstable when the cutoff reaches the sampling-resolution regime (\(K=8\)).
QSPADE replaces this all-or-nothing behavior by a graded transition from well-sampled modes to unresolved tail modes.}
\label{fig:expB-modes}
\end{figure}

Figure~\ref{fig:expB-modes} shows the consequence of this rank ambiguity
mode by mode.  Each hard
detector is binary: for example, the mode $m_6$, a perfectly valid sector
carrying $w_6\approx1.6\%$ of the average state (nine
expected training samples), is scored fully anomalous at $K=4$ and fully normal at $K=6$,
so the choice of $K$ flips entire sectors between the two verdicts and the
score carries no information about how marginal a sector is.  Once the cutoff
reaches the sampling resolution ($K=8$, where $\lambda_8\approx1/N$) the hard
score mislabels individual
modes erratically, because the trailing eigenvectors mix directions whose
eigenvalues are degenerate at the noise level, so increasing $K$ is not a
remedy.  Instead, the soft residual, with one $\alpha$ and no rank choice,
grades the same ladder monotonically: well-sampled modes score near zero,
boundary modes take intermediate values that track their rarity, and only the
modes at or below the sampling resolution, which no data-driven detector could
certify, score as anomalous.

\begin{figure}[t]
\centering
\includegraphics[width=\linewidth]{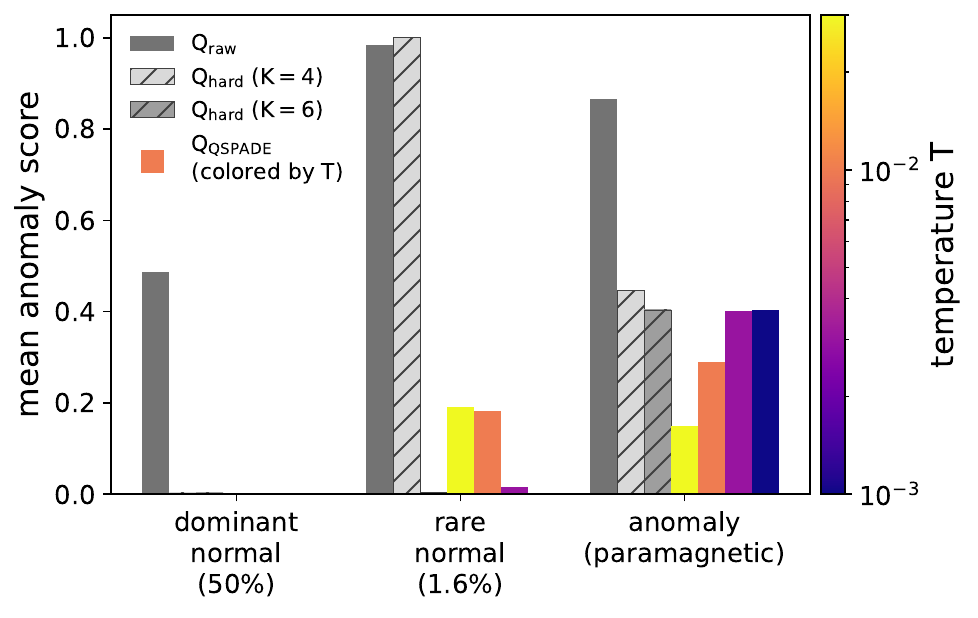}
\caption{Mean residual anomaly score for a dominant normal mode (\(50\%\)), a rare normal mode (\(1.6\%\)), and the paramagnetic anomaly class.
The density-weighted \(Q_{\rm raw}\) and the hard \(K=4\) projector score the rare valid mode above the anomaly, while the hard \(K=6\) projector gives it zero score and hides its rarity.
QSPADE assigns the rare mode an intermediate, temperature-dependent score; when \(T\) resolves the spectral tail, the rare valid mode remains below the anomaly. This shows soft scoring in QSPADE separates rarity within the normal support from genuine support mismatch.}
\label{fig:expB-sectors}
\end{figure}

\begin{figure}[t]
\centering
\includegraphics[width=\linewidth]{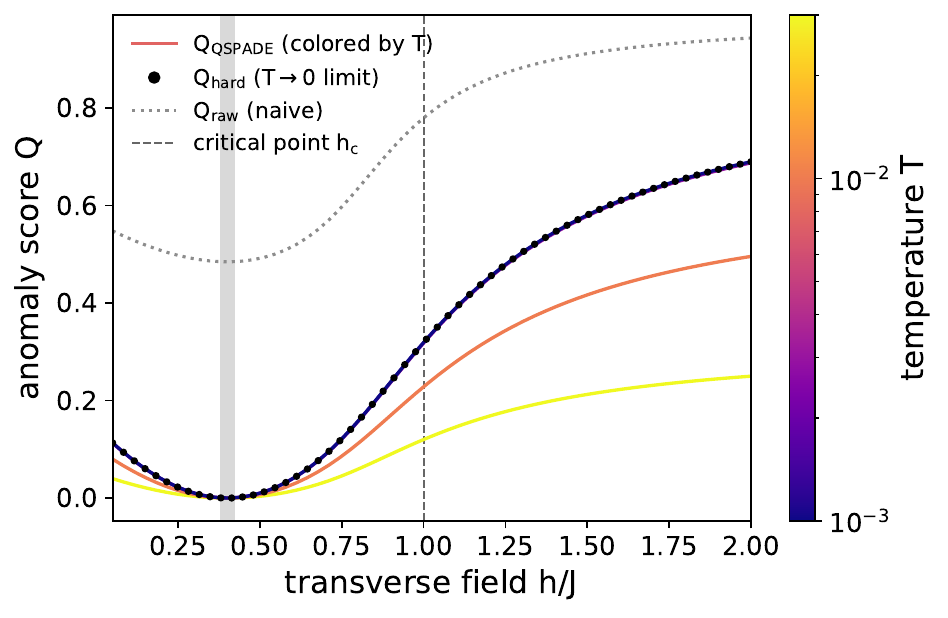}
\caption{QSPADE residual score of the TFIM ground state as the transverse field \(h\) is swept across the transition. 
The score is near zero at the normal training field \(h_A=0.4\) (grey line), stays low in the ordered phase, and rises through the critical point \(h_c=1\) (dashed) into the paramagnetic phase. 
Colored curves show different resolutions \(T\); as \(T\to0\) they approach the hard retained-rank projector \(K(\alpha)\) (black markers), while larger \(T\) gives a softer monitor. 
The density-weighted \(Q_{\rm raw}\) (dotted) is offset upward even at the normal field.}
\label{fig:expB-sweep}
\end{figure}

The rare-but-valid sector, the $1.6\%$ mode $m_6$ above, makes the
comparison concrete
(Fig.~\ref{fig:expB-sectors}).  The density-weighted $Q_{\rm raw}$ assigns it a mean score of $0.99$, above the genuine near-critical anomalies
at $0.86$; the hard projector at $K=4$ commits the same
inversion ($1.00$ against $0.45$), while at $K=6$, whose cutoff retains exactly the first six modes, it absorbs the sector
silently, with no trace of its rarity.  The calibrated $Q^f_\alpha$ keeps the
rare sector low ($0.02$ at $T=3\times10^{-3}$) and clearly below the anomaly
(at $0.40$), provided the resolution resolves the tail: for $T$
much larger than the spacing of the tail eigenvalues the detector over-retains
the tail and this contrast degrades.

This finite resolution and the smooth
occupation profile are precisely what the spectral-threshold measurement of Section~\ref{subsec:qumode-threshold} produces, fixed by one retained-mass
target rather than a discrete rank ladder, and the hard projector at the
retained-mass rank is recovered as its $T\to0$ endpoint.

Sweeping the transverse field across the critical point displays the whole soft
family at once (Fig.~\ref{fig:expB-sweep}).  The QSPADE residual is smallest
at the normal training field $h_A=0.4$, stays low throughout the ordered phase
($h<h_c$), and rises through the transition into the paramagnetic
phase ($h>h_c$), so anomaly detection tracks the quantum phase boundary.  The
family interpolates smoothly between a permissive
large-$T$ monitor and the hard projector recovered as $T\to0$.

\section{Conclusion}
\label{sec:conclusion}
We introduced QSPADE, a measurement-based framework for PCA-style anomaly detection with quantum data. 
The main contribution is to replace the reconstruction-first PCA workflow, including forming a covariance or Gram matrix, recovering principal eigenvectors, and selecting a hard rank, with directly getting an anomaly score from the average normal state \(C\). 
For the anomaly score, the residual \(Q\) score is the basic one-class detector and requires only retained-mass calibration, while the Hotelling-type \(\mathcal T^2\) score can be added when component-resolved leverage information is needed.

The soft spectral threshold is what makes this formulation different from standard PCA. 
Instead of accepting or rejecting whole eigenspaces by a discrete rank cutoff, QSPADE lets boundary spectral components contribute fractionally, with the hard-projector PCA score recovered in the zero-temperature limit \(T\to0\). 
This soft weighting reduces the sensitivity to eigenvalue crossings or small spectral gaps near the cutoff, since near-threshold modes are not forced into a binary rank decision.
It also avoids the need to resolve adjacent near-cutoff eigenvalues solely to decide a hard rank.
As a result, the default residual detector has dimension-independent measurement-shot complexity, while the optional leverage calibration is confined to the finite-window overhead \(\widetilde O(Ld)\).

The numerical experiments demonstrate the two intended roles of the framework. 
On encoded classical data, QSPADE reproduces kernel-PCA-like nonlinear decision regions. 
On quantum-native transverse-field Ising data, it detects the change across the phase transition without specifying an order parameter. 
These results support QSPADE as both a quantum-kernel anomaly detector and a data-driven monitor for quantum native systems.

Our results also raise concrete implementation questions. A natural next step is a small-scale hardware demonstration of QSPADE. 
The qumode construction is well matched to Gaussian responses and spectral sweeps, but the required qumode--qubit coupling may make the Hadamard-test route the more immediate option on qubit devices. 
At the circuit level, density matrix exponentiation gives a standard sample-access primitive, but tailored block-encodings or compiled simulations of \(C\) may reduce the overhead, which deserves further exploration. 
Finally, the same calibrated support viewpoint could be extended to multi-class quantum state discrimination, with separate learned spectral supports for different classes.

\section*{Acknowledgements}
MMW acknowledges support from the National Science Foundation under grant no.~2611810.
NL acknowledges funding from the
Science and Technology Commission of Shanghai Municipality (STCSM)
grant no.~24LZ1401200 (21JC1402900), NSFC grants no.~12471411 and
no.~12341104, the Shanghai Jiao Tong University 2030 Initiative, the Shanghai Pilot Program for Basic Research, 
and the Fundamental Research Funds for the Central Universities.

\bibliography{qspade_refs}

\appendix

\section{Proof of Proposition~\ref{prop:hard-support}}
\label{app:proof-hard-support}
\begin{proof}
For any effect $0\preceq M\preceq I$, let $Q_M$ be the orthogonal projector onto the range of $M$.  Since $0\preceq M\preceq Q_M$ and $\rank(Q_M)=\rank(M)\le K$,
\begin{equation}
    \Tr(MC)\le\Tr(CQ_M).
\end{equation}
Ky Fan's maximum principle~\cite{Bhatia1997} gives
\begin{equation}
    \Tr(CQ_M)\le\sum_{j=1}^K\lambda_j.
\end{equation}
Equality is achieved by $M=P_K$.  If $\lambda_K>\lambda_{K+1}$, the top-$K$ spectral subspace is unique.
\end{proof}

\section{Proof of Theorem~\ref{thm:rsd}}
\label{app:proof-rsd}

\begin{proof}
First, we establish various properties of the generator $\phi_f$. By the fundamental theorem of calculus and the inverse-function theorem, its derivatives are
\begin{equation}
    \phi_f'(m)=f^{-1}(m),
    \qquad
    \phi_f''(m)=\frac{1}{f'(f^{-1}(m))}>0.
\end{equation}
Since $f'>0$, the inverse $f^{-1}$ is strictly increasing, making $\phi_f$ strictly convex on $(0,1)$. Integrability gives a finite continuous extension at the endpoints. The definition yields $\phi_f(0)=0$, and the normalization condition in Definition~\ref{def:admissible} ensures $\phi_f(1)=\int_0^1f^{-1}(u)\,du=0$.

Consequently, the objective function in~\eqref{eq:regularized-objective} is strictly convex on the operator interval because $M\mapsto\Tr\phi_f(M)$ is strictly convex and the remaining terms are linear. Since $f$ maps $\mathbb{R}$ onto $(0,1)$, the derivative $\phi_f'(m)=f^{-1}(m)$ tends to $-\infty$ as $m\downarrow0$ and to $+\infty$ as $m\uparrow1$. Thus, the optimum must lie in the open interval $0\prec M\prec I$. 

The Fr\'echet first-order condition for the objective is
\begin{equation}
    \mu I-C+T\phi_f'(M)=0.
\end{equation}
Rearranging the terms yields
\begin{equation}
    \phi_f'(M)=\frac{C-\mu I}{T}.
\end{equation}
Applying the inverse derivative $(\phi_f')^{-1}=f$ via functional calculus gives~\eqref{eq:RSD-filter}. Strict convexity guarantees that this optimal solution is unique.
\end{proof}

\section{Proof of Proposition~\ref{prop:sharp-resolution}}
\label{app:proof-sharp-resolution}

\begin{proof}
For each eigenvalue $\lambda_j$, the scalar argument $\left(\lambda_j-\mu\right)/T$ tends to $+\infty$, $-\infty$, or $0$ depending on whether $\lambda_j>\mu$, $\lambda_j<\mu$, or $\lambda_j=\mu$, respectively.  The result follows from the boundary conditions in Definition~\ref{def:admissible}.
\end{proof}

\section{Proof of Proposition~\ref{prop:sharp-QT2}}
\label{app:proof-sharp-QT2}

\begin{proof}
By Proposition~\ref{prop:sharp-resolution}, each tail detector $M_\ell$ converges to the hard projector onto the modes above threshold $\mu_\ell$.  
Rank-calibrated thresholds isolate consecutive nondegenerate eigenmodes in the sharp limit, so that $D_j=M_j-M_{j-1}\to\ketbra{u_j}$.  
Substitution into~\eqref{eq:retained-score} and~\eqref{eq:qsad-T2} gives the residual and leverage limits, because $p_\sigma(j)\to |a_j|^2$ and $v_j=\Tr(D_j C)\to\lambda_j$.
\end{proof}

\section{Response examples for qumode spectral-threshold measurement}
\label{app:response-examples}

This appendix records two concrete optimization-admissible responses used in
the qumode construction of Section~\ref{subsec:qumode-threshold}.  In both
cases the symmetry \(f(-x)=1-f(x)\) implies
\(\int_0^1 f^{-1}(m)\,dm=0\), so the centering convention in
Definition~\ref{def:admissible} is satisfied.

\subsection{Fermi--Dirac response}

The increasing Fermi--Dirac, or logistic, response is
\begin{equation}
    f_{\rm FD}(x)\coloneqq \frac{1}{1+e^{-x}}.
\end{equation}
Its quantile and generator are
\begin{align}
    f_{\rm FD}^{-1}(m) & =\log\frac{m}{1-m},
    \,\\
    \phi_{\rm FD}(m)
    & =
    m\log m+(1-m)\log(1-m).
\end{align}
Its derivative is the logistic probability density
\begin{equation}
    f_{\rm FD}'(x)
    =
    \frac{e^{-x}}{(1+e^{-x})^2}
    =
    \frac14\,\operatorname{sech}^2\!\left(\frac{x}{2}\right).
\end{equation}
Therefore the reflected position density in~\eqref{eq:reflected-density} is
\begin{equation}
\label{eq:fd-reflected-density}
    g_{{\rm FD},T_1}(q)
    =
    \frac{1}{4T_1}
    \operatorname{sech}^2\!\left(\frac{q}{2T_1}\right),
\end{equation}
and the corresponding control wavepacket is
\begin{equation}
\label{eq:fd-wavepacket}
    \ket{\psi_{{\rm FD},T_1}}
    =
    \int_{\mathbb R}
    \frac{1}{2\sqrt{T_1}}
    \operatorname{sech}\!\left(\frac{q}{2T_1}\right)
    \ket q\,dq.
\end{equation}
With \(T=T_1T_2\) and \(\mu=\beta T_2\), Proposition~\ref{prop:cv-threshold}
gives the effect
\begin{equation}
    E^{{\rm FD}}_{\beta,T_1,T_2}
    =
    \left[
    I+\exp\!\left(-\frac{C-\mu I}{T}\right)
    \right]^{-1}.
\end{equation}

\subsection{Gaussian response}

The Gaussian, or probit, response is the standard normal cumulative
distribution function
\begin{equation}
    f_{\rm G}(x)=\Phi(x)
    \coloneqq 
    \frac12\left[
    1+\operatorname{erf}\!\left(\frac{x}{\sqrt2}\right)
    \right]
\end{equation}
where $\mathrm{erf}(\cdot)$ denotes the error function.
Its derivative is the standard normal density
\begin{equation}
    f_{\rm G}'(x)
    =
    \frac{1}{\sqrt{2\pi}}e^{-x^2/2}.
\end{equation}
The quantile and generator are
\begin{align}
    f_{\rm G}^{-1}(m) & =\Phi^{-1}(m),
    \,\\
    \phi_{\rm G}(m)
    & =
    -\frac{1}{\sqrt{2\pi}}
    \exp\!\left[-\frac12\bigl(\Phi^{-1}(m)\bigr)^2\right].
\end{align}
The reflected position density is the Gaussian density
\begin{equation}
\label{eq:gaussian-reflected-density}
    g_{{\rm G},T_1}(q)
    =
    \frac{1}{\sqrt{2\pi}T_1}
    \exp\!\left(-\frac{q^2}{2T_1^2}\right),
\end{equation}
so the control wavepacket is the Gaussian qumode state
\begin{equation}
\label{eq:gaussian-wavepacket}
    \ket{\psi_{{\rm G},T_1}}
    =
    \int_{\mathbb R}
    \frac{1}{(2\pi T_1^2)^{1/4}}
    \exp\!\left(-\frac{q^2}{4T_1^2}\right)
    \ket q\,dq.
\end{equation}
Again \(T=T_1T_2\) and \(\mu=\beta T_2\), and the implemented effect is
\begin{align}
    E^{{\rm G}}_{\beta,T_1,T_2}
    & =
    \Phi\!\left(\frac{C-\mu I}{T}\right)\\
    & =
    \frac12\left[
    I+
    \operatorname{erf}\!\left(
    \frac{C-\mu I}{\sqrt2\,T}
    \right)
    \right].
\end{align}

\section{Proof of Proposition~\ref{prop:cv-threshold}}
\label{app:proof-cv-threshold}

\begin{proof}
Let $C=\sum_j\lambda_j\ketbra{u_j}$.  Under $e^{-i\hat p\lambda_j/T_2}$, the position distribution is shifted by $\lambda_j/T_2$.  The density conditioned on eigenmode $\ket{u_j}$ is given by
\begin{equation}
    g_{f,T_1}\!\left(q-\frac{\lambda_j}{T_2}\right).
\end{equation}
The tail probability is
\begin{align}
    & \int_\beta^\infty
    \frac1{T_1}f'\!\left(-\frac{q-\lambda_j/T_2}{T_1}\right)dq \notag \\
    &=\int_{-\infty}^{(\lambda_j/T_2-\beta)/T_1}f'(z)\,dz \\
    &=f\!\left(\frac{\lambda_j-\beta T_2}{T_1T_2}\right).
\end{align}
Functional calculus gives~\eqref{eq:cv-effect}.
\end{proof}

\section{Proof of Proposition~\ref{prop:control-shifted-estimator}}
\label{app:proof-control-shifted-estimator}

\begin{proof}
For a sampled $\omega$, write $\chi_\omega=|\chi_\omega|e^{i\varphi_\omega}$.  The identity shift factors as
\begin{equation}
\label{eq:qubit-shift-factorization}
    e^{it_\omega(C-\mu I)/T}
    =e^{-i\mu t_\omega/T}e^{it_\omega C/T}.
\end{equation}
A Hadamard test with controlled unitary $e^{it_\omega C/T}$ and control phase
\begin{equation}
    \theta_\omega(\mu)=\varphi_\omega-\mu t_\omega/T
\end{equation}
outputs a sign variable $Z_\omega\in\{-1,+1\}$ satisfying
\begin{equation}
    \EE[Z_\omega\mid\omega]
    =
    \Re\!\left[
    e^{i\theta_\omega(\mu)}
    \Tr\!\left(e^{it_\omega C/T}\rho\right)
    \right].
\end{equation}
Set $Y_\mu(\rho)=b_f+A_f|\chi_\omega|Z_\omega$.  Taking the expectation over the measurement outcome from the Hadamard test and over $\omega$ gives
\begin{align}
    & \EE Y_\mu(\rho) \notag \\
    &=b_f+A_f\EE_\omega
    \Re\!\left[
    \chi_\omega e^{-i\mu t_\omega/T}
    \Tr\!\left(e^{it_\omega C/T}\rho\right)
    \right] \\
    &=\Tr\!\left[\rho\,\widetilde f\!\left(\frac{C-\mu I}{T}\right)\right].
\end{align}
The uniform scalar approximation error implies an operator-norm error at most $\varepsilon_f$, which proves~\eqref{eq:general-estimator-bias}.  The range bound follows from $|\chi_\omega|\le1$ and $|Z_\omega|=1$.
\end{proof}

\section{Supplementary numerical results}
\label{app:numerics}

This appendix collects supporting material for the experiments of
Section~\ref{sec:numerics}: the feature maps and simulation parameters, the
resolution--shot trade-off, and the kernel and geometry comparisons.

\subsection{Experiments setup}
\label{app:setup}

In the classical-data experiment of Section~\ref{subsec:expA}, each classical
point $x=(x_1,x_2)\in[0,1]^2$ is mapped to a feature state
$\ket{\phi(x)}$ by one of the maps below; we write $\tilde x_i=2x_i-1$.  All
maps produce a proper pure state of unit Euclidean norm: the proportionality sign
in $\phi\propto(\cdots)$ denotes this normalization (the amplitude vector is
divided by its norm).  The Fourier map is used in the main text as a fixed
representative; the comparison across feature maps and classical kernels appears
in Fig.~\ref{fig:expA-kernels}.
\begin{itemize}
\item \emph{Degree-two polynomial}:\\
  $\phi\propto(1,\tilde x_1,\tilde x_2,\tilde x_1^2,\tilde x_1\tilde x_2,\tilde x_2^2)$.
\item \emph{Fourier} (four qubits, $16$-dimensional): the tensor product over the
  two coordinates of the per-axis features
  $(\cos\pi t,\cos 2\pi t,\sin\pi t,\sin 2\pi t)$, $t\in\{x_1,x_2\}$.
\item \emph{Gaussian random features} (four qubits): $\phi\propto\cos(Wx+b)$ with
  $W_{ij}\sim\mathcal N(0,2\xi)$ and $b_i\sim\mathrm{Unif}[0,2\pi)$, a
  $16$-dimensional Monte-Carlo approximation of the RBF kernel
  $e^{-\xi\lVert x-x'\rVert^{2}}$.
\end{itemize}
The classical kernel-PCA
baselines use the RBF kernel $e^{-\xi\lVert x-x'\rVert^{2}}$ (the baseline of
the main text), the Laplacian
kernel $e^{-\xi_1\lVert x-x'\rVert_1}$ with
$\xi_1=\big(\mathrm{median}_{i,j}\lVert x_i-x_j\rVert_1\big)^{-1}$, and the
cubic polynomial kernel $(1+x\cdot x')^3$.

\begin{figure}[t]
\centering
\includegraphics[width=\linewidth]{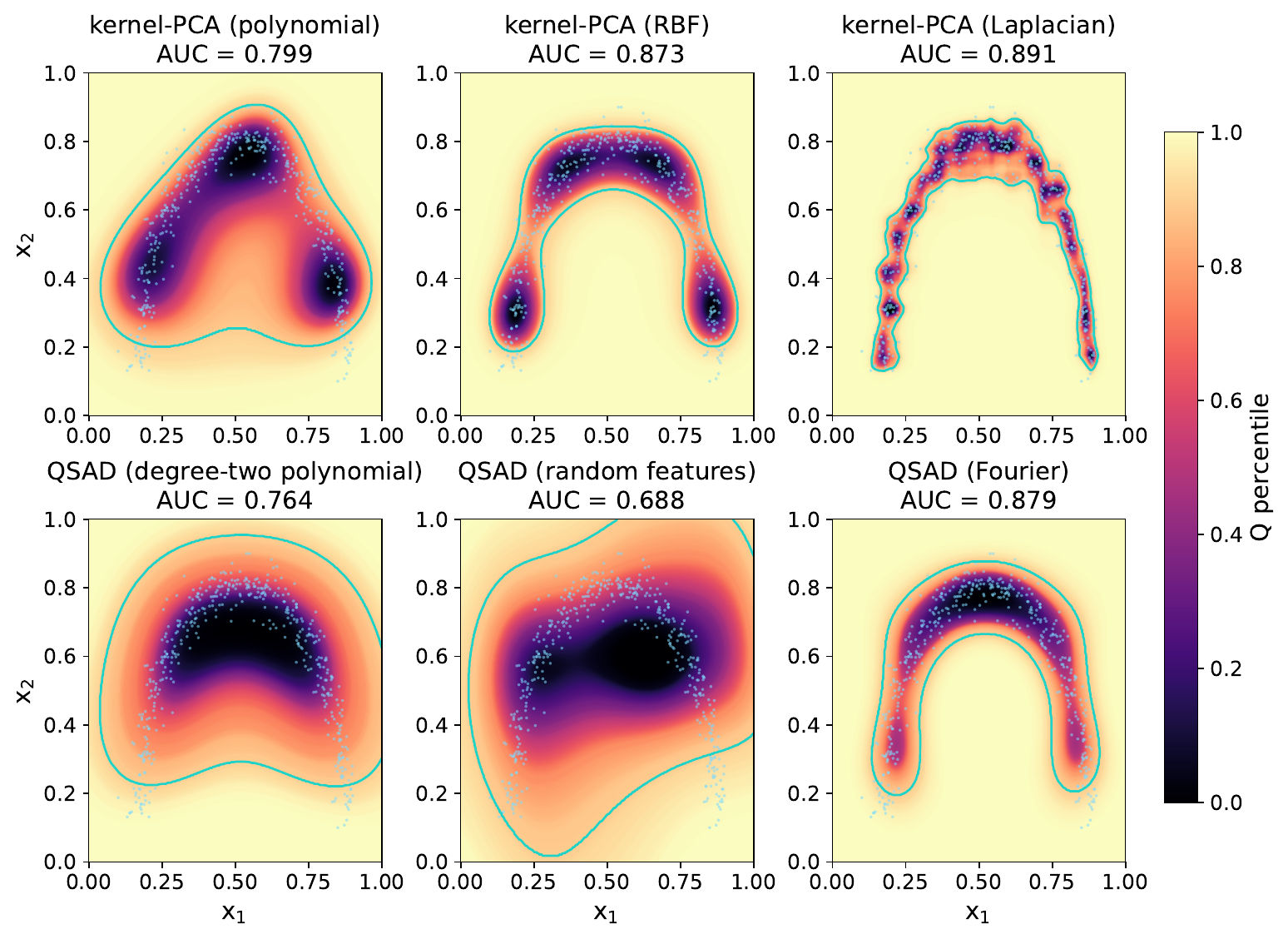}
\caption{Residual
percentile on the moon dataset for three classical kernel-PCA baselines (top row) and
QSPADE with three feature maps (bottom row), with the one-class detection AUC
(against a uniform background) in each panel title.}
\label{fig:expA-kernels}
\end{figure}

For the classical-data experiment, the normal clouds are a single moon (and, for
the geometry sweep, an elongated Gaussian blob and a ring), each with
$N=450$--$500$ points rescaled into $[0,1]^2$.  The RBF kernel-PCA bandwidth is set
by the median heuristic
$\xi=\big(2\,\mathrm{median}_{i,j}\lVert x_i-x_j\rVert^{2}\big)^{-1}$ and its
retained rank by the same explained-variance target as QSPADE's retained mass;
the QSPADE detector uses $\alpha=0.88$, $T=0.1$ in the three-method grid of
Fig.~\ref{fig:expA},
$\alpha=0.88$, $T=0.035$ in the kernel comparison of
Fig.~\ref{fig:expA-kernels}, and $\alpha=0.90$, $T=0.05$ in the geometry sweep
underlying Table~\ref{tab:expA}.  

For the quantum-native experiment, the normal ladder and anomalies are as in Section~\ref{subsec:expB} ($M=12$ modes
with weights $\propto2^{-j}$, $N=600$ training samples, perturbation amplitude
$0.05$), with $\alpha=0.99$ and, unless a family of resolutions
$T\in[10^{-3},3\times10^{-2}]$ is shown,
$T=3\times10^{-3}$.  The expected per-mode training counts $w_jN$ fall from
$300$ at $j=1$ to nine at $j=6$, about one at $j=9$, and below one for the last
modes.  Sector and per-mode scores are means over $300$ states sampled as in
training (a single normal mode per sector, the $0.7/0.2/0.1$ paramagnetic
mixture for the anomaly class); the detection AUC compares $600$
normal-mixture samples against $300$ anomalous ones; and the field sweep of
Fig.~\ref{fig:expB-sweep} tracks the single ground state at each field $h$.

All scores are converted
to normal-calibrated percentiles through the empirical cumulative distribution of
the normal scores; detection AUC uses uniform-background anomalies (classical) or
paramagnetic-phase states (quantum-native).  Centering of the training and
test data follows each
method's standard convention: linear PCA centers the inputs with the training-set mean, the kernel-PCA baselines center in feature space through the standard double-centering of the Gram matrix, and the feature maps act on the mean-centered data.

\begin{figure}[t]
\centering
\includegraphics[width=\linewidth]{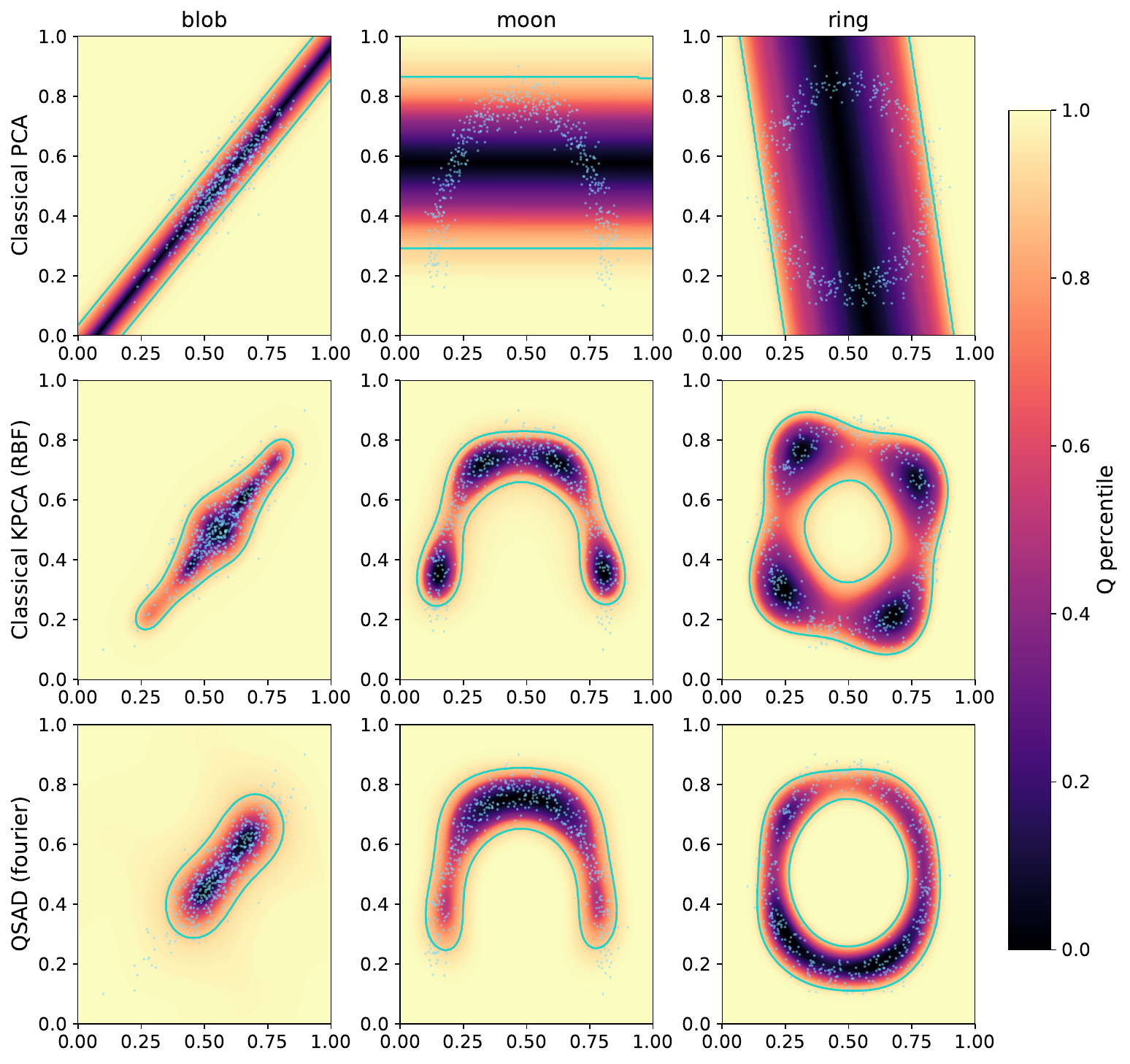}
\caption{Residual-percentile decision
regions for linear PCA (top), RBF kernel-PCA (middle), and QSPADE (bottom) on the
blob, moon, and ring.  Linear PCA cannot follow the curved or non-convex
supports; both nonlinear monitors do.}
\label{fig:expA-datasets}
\end{figure}

Figure~\ref{fig:expA-kernels} places the classical and the quantum side of
the kernel-PCA correspondence of Section~\ref{subsec:expA} side by side
on the moon: classical kernel-PCA with polynomial, RBF, and Laplacian kernels in
the top row, and QSPADE with the degree-two polynomial, random-feature, and
Fourier maps in the bottom row, under the same calibration and with the
one-class AUC in each panel.  The columns pair related kernels.  The two
polynomial methods produce similarly broad acceptance regions, the weakest
separation among the classical baselines; the random-feature map is a
sixteen-dimensional approximation of the exact RBF kernel above it, visibly
coarser at this encoding size and the weakest panel on the QSPADE side; and the
Laplacian and Fourier panels give the tightest boundaries on each
side. Because the classical-data experiment is a
consistency check rather than an accuracy claim, the encoding is not tuned for
performance; the Fourier map serves as a fixed representative in the main text,
with the comparison across kernels and maps reported here.
Figure~\ref{fig:expA-datasets} shows the three-method comparison
across the three synthetic classical datasets, the blob, moon, and ring
geometries, that underlies Table~\ref{tab:expA}:
linear PCA cannot follow the curved or non-convex supports, most strikingly the
ring, where its straight acceptance band runs through the empty centre, whereas
the RBF kernel-PCA baseline and QSPADE both track the manifold.

\subsection{Resolution--shot trade-off (quantum-native experiment)}

The QSPADE score is read off finite measurement shots: the acceptance
$\bra{\psi}M^f_{\mu_\alpha,T}\ket{\psi}$ is estimated from $m$ binary
accept/reject outcomes, so the empirical residual carries binomial sampling
noise of order $m^{-1/2}$.  The data and calibration are those of the
quantum-native experiment of Section~\ref{subsec:expB}: the detector is trained
on the graded normal ladder of the ordered TFIM at $h_A=0.4$ and calibrated at
$\alpha=0.99$, and the AUC separates fresh normal-mixture samples from
paramagnetic anomalies at $h_B=1.2$.  Figure~\ref{fig:expB-auc} plots the detection ROC-AUC
against $T$ for several shot budgets, the $\infty$-shot curve being the exact
score.  A sharper detector (small $T$) places the normal and anomalous
acceptances near the well-separated extremes $0$ and $1$, where the per-state
variance is small, and stays robust down to a few tens of shots; a softer
detector (large $T$) compresses the acceptances, so the normal--anomaly gap shrinks relative to the shot noise and the AUC
degrades. This is the cost
side of the resolution knob: the larger $T$ that
yields a more graded, information-rich score also demands more measurement shots
to resolve.

\begin{figure}
\centering
\includegraphics[width=0.9\linewidth]{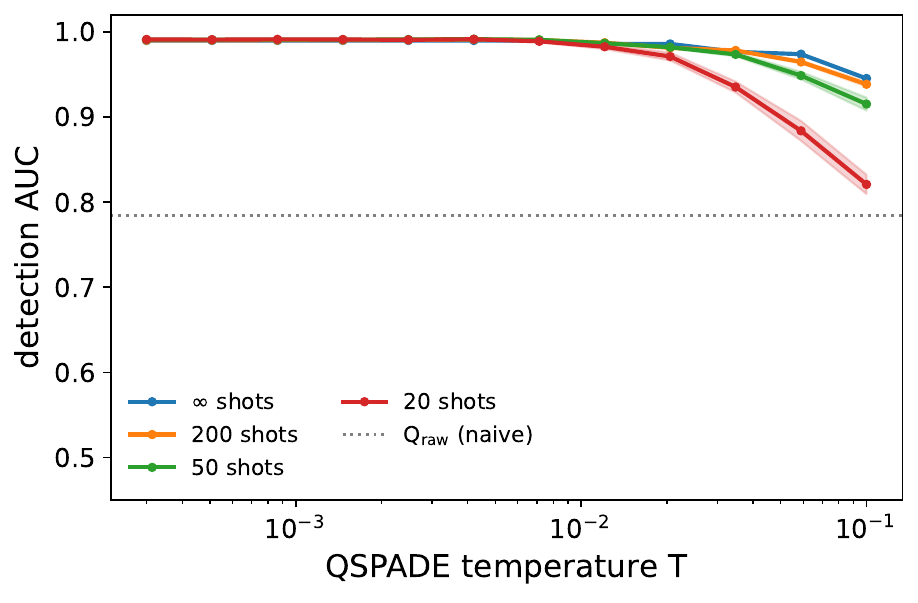}
\caption{Detection ROC-AUC versus resolution $T$ at infinite and finite
measurement-shot budgets for the quantum-native experiment ($n=10$):
normal-mixture samples from the graded ladder at $h_A=0.4$ versus paramagnetic
anomalies at $h_B=1.2$.  Smaller $T$ is markedly more robust to shot
noise; the dotted line is the density-weighted $Q_{\rm raw}$ baseline.}
\label{fig:expB-auc}
\end{figure}

\end{document}